\pdfoutput=1
\documentclass[preprint,12pt,3p]{elsarticle}
 
\usepackage{microtype}

\usepackage[utf8]{inputenc}
\usepackage[english]{babel}
\usepackage{amssymb}

\usepackage{graphicx,amssymb,amsmath,amsthm,xcolor}
\usepackage[utf8]{inputenc}

\usepackage{sidecap}  
\usepackage{comment}
\usepackage{caption}

\usepackage{algorithmic}
\usepackage{subfig}
\usepackage{enumitem}
\usepackage[export]{adjustbox}
\usepackage{booktabs}
\usepackage[algoruled]{algorithm2e}
\SetKwInOut{Input}{input}
\SetKwInOut{Output}{output}
\SetKw{KwT}{True} 
\SetKw{KwF}{False}
\SetKw{KwOr}{or}
\SetKw{KwAnd}{and}
\DontPrintSemicolon
\journal{Theoretical Computer Science}

\newtheorem{obs}{Observation}
\newtheorem{lemma}{Lemma}
\newtheorem{theorem}{Theorem}

\newtheorem{corollary}{Corollary}[theorem]

\graphicspath{{figures/}}

\newcommand{\etal} {{\it et al.}\xspace} 

\newcommand{\mama}{{\tt MaxMaxArea}\xspace}
\newcommand{\mami}{{\tt MaxMinArea}\xspace}
\newcommand{\mima}{{\tt MinMaxArea}\xspace}
\newcommand{\mimi}{{\tt MinMinArea}\xspace}

\newcommand{\niceremark}[3]{\textcolor{name}{\textsc{\textbf{#1 #2: }}}\textcolor{red}{\textsf{#3}}\xspace}

\definecolor{name}{rgb}{0.5,0.0,0.0}
\newcommand{\maarten}[2][says]{\niceremark{Maarten}{#1}{#2}}

\renewcommand{\paragraph}[1]{\smallskip\noindent{\bfseries #1.}}

\begin{document}
\begin{frontmatter}

\title{Largest and Smallest Area Triangles on  Imprecise Points\footnote{A preliminary version of this paper was presented at EuroCG2017.}}

\author[1]{Vahideh Keikha}
\ead{va.keikha@aut.ac.ir}
\author[2]{Maarten L\"offler}
\ead{m.loffler@uu.nl}
\author[1]{Ali  Mohades}
\ead{mohades@aut.ac.ir}
\address[1]{Department of Mathematics and Computer Science, Amirkabir University of  Technology, Tehran, Iran}
\address[2]{Department of Information and Computing Sciences, Utrecht University, Utrecht, The Netherlands}



\begin{abstract}

Assume we are given a set of parallel line segments in the plane, and we wish to place a point on each line segment such that the  resulting point set  maximizes or minimizes the area of the largest or smallest triangle in the set. 
We analyze the complexity of the four resulting computational problems, and we show that three of them admit polynomial-time algorithms, while the fourth is NP-hard.
Specifically, we show that
maximizing the largest triangle can be done in $O(n^2)$ time (or in $O(n \log n)$ time for unit segments);
minimizing the largest triangle can be done in $O(n^2 \log n)$ time;
maximizing the smallest triangle is NP-hard; but
minimizing the smallest triangle can be done in $O(n^2)$ time. We also discuss to what extent our results can be generalized to polygons with $k>3$ sides.
\end{abstract}
\end{frontmatter}


\section{Introduction}
In this paper we study two classical problems in  computational geometry in an imprecise context. 
Given a set $P$ of $n$ points in the plane,  
let the {\em largest-area triangle} $T_{max}(P)$ and the {\em smallest-area triangle} $T_{min}(P)$ be defined by three points of $P$ that form the triangle with the largest or smallest area, respectively. 
When $P$ is  uncertain, the areas of $T_{max}$ and $T_{min}$ are also uncertain. Our aim is to compute tight bounds on these areas given bounds on the locations of the points in $P$. As a natural extension, we also study all the above questions for  $k$-gons instead of triangles.

\paragraph{Motivation}
Data uncertainty is paramount in present-day geometric computation.
Many different ways to model locational uncertainty have been introduced over the past decades, and can be mainly categorized by whether the uncertainty is {\em discrete} or {\em continuous}~\cite {jlp-gcip-11}, and whether we assume the uncertainty is governed by an underlying probability distribution or not~\cite {surimost,agarwalconvex}.
In this paper, we assume the uncertainty in each point is captured by a continuous set of possible locations; we call such a set an {\em imprecise} point.
This model can be traced back to early attempts to create robust geometric algorithms in the 80s~\cite {Salesin}, and has attracted considerable attention since~\cite {AHN2013253,jup,kl-gmiphd-09,lofflerphdthesis,Myers:2010:UGD:1839778.1839801,farnaz}.
Nagai and Tokura~\cite {nt-teb-00} first introduced the idea of analyzing the computational complexity of computing tight error bounds on an output value based on a set of imprecise points, and L\"offler and van Kreveld formalized this notion to calculate bounds on the area of the convex hull~\cite{39}.

The special case of using vertical line segments as uncertainty regions has received special attention in the literature.
First, it is a natural first step towards general 2- or 3-dimensional uncertainty regions. Many geometric problems where first studied on segments and later generalized to squares or disks, and this is also true for imprecision~\cite {mountproximity,58,39}.
However, 1-dimensional uncertainty already occurs naturally in several application areas.
For instance laser scanners output points on a known line but the distance between the scanner and the scanned object has an error;
this is especially significant when the distance is large, such as in LIDAR (Light Detection and Ranging) data~\cite{tasdizen2003feature} leading to distinct geometric challenges~\cite {gls-si15dt-10}.

The same geometric problems also show up in, and are sometimes studied from the point of view of, different applications.
In imaging, the problem of whether a set of vertical (or horizontal) scan lines are stabbed by a geometric object has been studied extensively several decades ago.
When the object is to test whether the set of segments may be {\em stabbed} by a convex polygon~\cite {gsch},
this is equivalent to asking whether a set of imprecise points, modeled as vertical segments, could possibly be in convex position~\cite {56,kgonsolved}.
Computing the largest inscribed/inscribing polygons plays an important role in heuristic motion planning~\cite{Fleischer:1992}; robots increasingly operate in uncertain environments.



\paragraph {Related work}
There is a large body of research on the existence and computation of {\em empty} $k$-gons in a set of points~\cite{41,43}; the best known time for finding such a $k$-gon for arbitrary $k$ is $O(T(n))$, where $T(n)$ is the number of empty triangles in the set of points---we know this value can vary from $\Omega(n^2)$   to $O(n^3)$ \cite{42}. \par

There has been some work  which focuses on constant values of $k$.  
The special case $k=2$ is the classical problem of finding the diameter of a given set of points.  Shamos  presented an algorithm for the diameter problem which can find the diameter in linear time, if the convex hull of the points is given~\cite{shd}. 

Dobkin and Snyder \cite{45} claimed  a linear-time algorithm for finding the largest-area triangle inscribed in a convex polygon. Their claim has recently been shown to be incorrect by Keikha \etal~\cite{kluv}; see also~\cite {jin,kallus}.
There exists, however, another linear-time algorithm for this problem~\cite{chandran}, originally intended to solve the parallel version of the problem. 


Boyce \textit{et al}. \cite{48} presented a dynamic programming algorithm for the problem of finding the largest possible area and perimeter convex $k$-gon on a given set $P$ of  $n$ points in 
$O(kn \log n+ n \log^2 n)$ time and linear space, that  
 Aggarwal \textit{et al}. \cite{msearch} improved to $O(kn+ n \log n)$ by using a matrix search method.
Both algorithms still rely on the correctness of the Dobkin and Snyder algorithm for triangles~\cite {45}, and hence, also fail by the analysis in~\cite{kluv}.
There is also an $\Omega(n \log n)$ time lower bound for finding the maximum possible area and perimeter inscribed $k$-gon \cite{52}. 

The problem of finding the {\em smallest} possible area and perimeter $k$-gon has received considerable attention as well. 
Dobkin \textit{et al}. presented an $O(k^2n \log n+k^5 n)$ time algorithm \cite{newdob} for finding minimum perimeter $k$-gons. Their algorithm was improved upon by Aggarwal \textit{et al}. to $O(n \log n +k^4 n)$ time \cite{49}.  
Eppstein \textit{et al}. \cite{50} studied three problems: finding the smallest possible $k$-gon, finding the smallest empty $k$-gon, and finding the smallest possible convex polygon on exactly $k$ points, where the smallest means the smallest possible area or perimeter. They presented a dynamic programming approach for these problems in   $O(kn^3)$ time and $O(kn^2)$ space, that can also solve the  maximization version of the problem as well as  some other related problems. Afterwards, Eppstein  \cite{51} presented an algorithm that runs in $O(n^2 \log n)$ time and $O(n \log n)$ space for  constant values of $k$. 

Finally, we mention a large body of related research, such as  stabbing problems or convex transversal problems, and proximity problems. In the two first aforementioned areas, the general problem is that we are given a set of geometric input objects and we want to find another object which intersects with all or most of the given input objects, such that  some measure on that object is optimized \cite{56,54,58,gsch,kgonsolved,39}.
 Also the input of such problems can be considered as a set of colored points, e.g., Daescu \textit{et al.} \cite{72} studied the following related problem: we are given a set of $n$ points with $k <n $ colors, and we want to find the convex polygon with  the smallest possible perimeter such that the polygon  covers at least one of the given colors. They presented an $\sqrt{2}$- approximation algorithm  with $O(n^2)$ time for this problem,  and proved that this problem is NP-hard if $k$ is a part of the input. 
 
  Similarly, there are many studies in  proximity problems,  where the general question is the following: given a set of $n$
 points in the plane $P = \{p_1, . . . , p_n\}$, for each point $p_i$ find a pair $p_j , p_k$, where $ i \neq j\neq k$, such that a defined measure $\alpha$ on the triplet  $p_i,p_j,p_k$ is maximized or  minimized. In ~\cite{mukhopadhyay2006all} the authors studied the problem of computing the maximum value of $\alpha$, where  $\alpha$ defined by the  distance of each point $p_i \in P$ from a segment $p_jp_k$, where the distance from a point $p_i$ to a segment 
 is the minimum distance from $p_i$ to this segment. Their algorithm runs in $O(nh+n \log n)$, where $h$ is the number of vertices on the convex hull of $P$. Their running time improved to $O(n \log n)$ ~\cite{drysdale2008nlogn}.  
 Recently the problem of computing all the largest/smallest area/perimeter triangles with a vertex at $p_i \in P$ for $i=1,...,n$ was studied in~\cite{mukhopadhyay2013all} (see also the references there), where  the presented running times for the largest/smallest area triangle problems were quadratic in the worst case (we also achieve similar running times with uncertain input). 

It is natural to  ask how uncertainty  of data affects the solutions of those problems. 
Motivated by this, we are  interested in computing some lower bound and upper bound on the area of the smallest and largest $k$-gons with  vertices on a given set of imprecise points modeled as parallel line segments.

  In the imprecise context, L\"offler and van Kreveld studied the diameter problem ($k=2$) on a given set of imprecise points modeled as squares or as disks, where the problem is choosing a point in each square or disk such that the diameter of the resulting point set is
as large or as small as possible. 
They presented an $O(n \log n)$ time algorithm for finding the maximum/minimum possible diameter on a given set of squares, and presented an $O(n \log n)$ time algorithm for the largest possible diameter, and an approximation scheme with $O(n^{ \frac{3\pi}{ \sqrt{ \varepsilon} }} )$ time  for the smallest possible diameter on a given set of disks.

These authors also   computed some lower and upper bounds on the smallest/largest area/perimeter convex hull, smallest/largest area bounding box, smallest/largest smallest enclosing circle, smallest/largest width and closest pair, where the input was imprecise, and modeled by  convex regions which include line segments, squares or disks~\cite{39,diam}. 
The running times of the presented algorithms on the convex hull problem vary from $O(n \log n)$ to $O(n^{13})$. Their results on computing some bounds on the maximum area  convex hull were later improved upon by Ju \textit{et al.}~\cite{jup}.





\begin{figure}
	\centering
\includegraphics{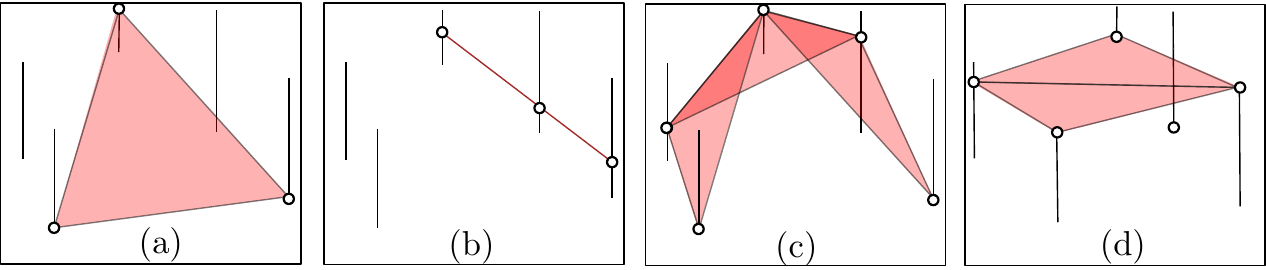}

\caption{Problem definition and optimal solutions. (a) \mama: the largest possible area  triangle. (b) \mimi: the smallest possible area  triangle, here degenerate.  (c) \mami: the largest smallest-area triangle, here determined by three triangles simultaneously. (d) \mima: the smallest largest-area triangle, here determined by two triangles simultaneously.}
\label{mnmxmxmn}

\end{figure}

\paragraph{Contribution}  \label{sec:problemdef}
In this paper, we consider the problems of computing the largest-area triangle and smallest-area triangle under data imprecision.
 We are given a set  $L=\{l_1, l_2,\ldots, l_n\}$ of imprecise points modeled as disjoint parallel line segments, that is, every segment $l_i$ contains exactly one point $p_i \in l_i$. This gives a point set $P=\{p_1, p_2,\ldots, p_n\}$, and we want to find the largest-area triangle or smallest-area triangle  in $P$, $T_{max}$ and $T_{min}$. But because $L$ is  a set of imprecise  points, we do not know where $P$ is, and the areas of these triangles could have different possible values for each instance $P$.
 We are interested in computing a tight lower and upper bound on these values. Hence, the problem becomes to place a point on each line segment such that the  resulting point set maximizes or minimizes the size of the largest or smallest possible area triangle.  Therefore four different problems need to  be considered (refer to Figure~\ref{mnmxmxmn}).

\begin{itemize}
\item \mama \quad What is the largest possible area of $T_{max}$?
\item \mimi \quad What is the smallest possible area of $T_{min}$? 
\item \mami \quad What is the largest possible area of $T_{min}$? 
\item \mima \quad What is the smallest possible area of $T_{max}$?

\end{itemize}


 \paragraph{Results}
   We obtain the following results for triangles.
  
  \begin {itemize} [noitemsep]
\item For a given set of equal length parallel line segments, \mama can be solved in  $O(n \log n)$ time\footnote {In a preliminary version of this paper, we claimed a faster method, which relied on the correctness of the Dobkin and Snyder  algorithm for finding the largest-area inscribed triangle~\cite{45}. That paper has since been shown incorrect~\cite{kluv}, and our present results reflect this. } (Section~\ref{sec:eqlen}).
\item  For arbitrary length parallel line segments, \mama can be solved in  $O(n^2)$ time (Section~\ref{sec:alen}).
\item \mimi can be solved  in  $O(n^2)$ time (Section~\ref{sec:minmin}).
\item \mami is NP-hard (Section~\ref{sec:maxmin}). 
\item  For  arbitrary length parallel line segments with fixed points as the leftmost and rightmost segments, \mima can be solved  in  $O(n \log n)$ time (Section~\ref{sec:fixed}).
\item  For  arbitrary length  parallel line segments \mima can be solved  in  $O(n^2 \log n)$ time (Section~\ref{sec:GC}).
\end{itemize}

   We also discuss to what extent our results can be generalized to polygons with $k>3$ sides.\footnote{As said before, for the case $k=2$, the problem becomes computing the smallest/largest diameter of $L$ which is studied in~\cite{diam}. } As we discuss in Section~\ref{sec:generalk}, not all problems are well-posed anymore, but we can ask for the possible sizes of the largest convex shape with {\em at most $k$ vertices}. We obtain the following results.

  \begin {itemize} [noitemsep]

\item  \mama for {\em at most} $k$ points can be solved in $O(kn^3)$ time (Section~\ref{sec:kmama}).
\item  \mima for {\em at most} $k$ points can be solved  in  $O(k n^8 \log n)$ time (Section~\ref{sec:kmima}).
\end{itemize}
\begin{figure}
	\includegraphics{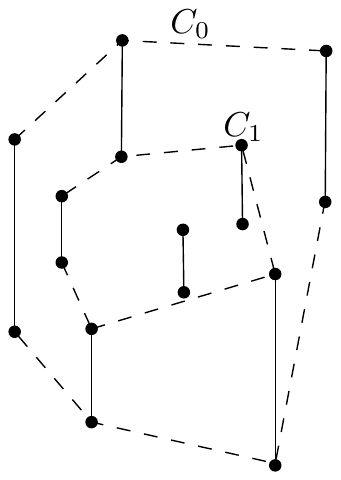}
	\centering
	\caption{$C_0$ and $C_1$ on  a given set of line segments.}
	\label{csdef}
\end{figure}

\paragraph{Definitions}
Without loss of generality, assume the parallel segments in $L$ are vertical. 
 Let $Z=\{l_1^{-}, l_1^{+}, l_2^{-}, l_2^{+},\ldots, l_n^{-}, l_n^{+}\}$ be the set of all endpoints of $L$, where $l_i^{+}$ denotes the upper endpoint of $l_i$, and  $l_i^{-}$ denotes the lower endpoint of  $l_i$. Let $CH(Z)$ and $\partial CH(Z)$ denote the convex hull and the  boundary of the convex hull of $Z$, respectively.
We define $C_0=CH(Z)$ as the convex hull of $Z$, and $C_1=CH(Z \setminus  \partial CH(C_0))$, that is $C_0$ and $C_1$ are the first two layers of the onion decomposition of $Z$, 
as shown in Figure~\ref{csdef}. 
  A \emph{true} object is an object such that all its vertices lie on distinct line segments in $L$. In particular we use the terms {\em true triangle}, {\em true convex hull}, {\em true chord} and {\em true edge} throughout the paper.
 Note that an optimal solution to any of our problems is always a true object.

 Boyce \emph{et al.}~\cite{45} defined a \emph{rooted} triangle (or more generally, a rooted polygon) as a triangle with one of its vertices fixed at a given point (in a context where the rest of the vertices are to be chosen from a fixed set of candidates).
 Here, we define a \emph{root} as a given point on a specific line segment in $L$. In this case we throw out the remainder of the root's region and try to find the other two vertices in the remaining $n-1$ regions. 
 For a given point $a$ on some line segment, we denote this segment by $l_a$. Also  $l_l$ and $l_r$ denote the leftmost and  rightmost line segments, respectively.

 \begin{figure}
	\centering
	\includegraphics{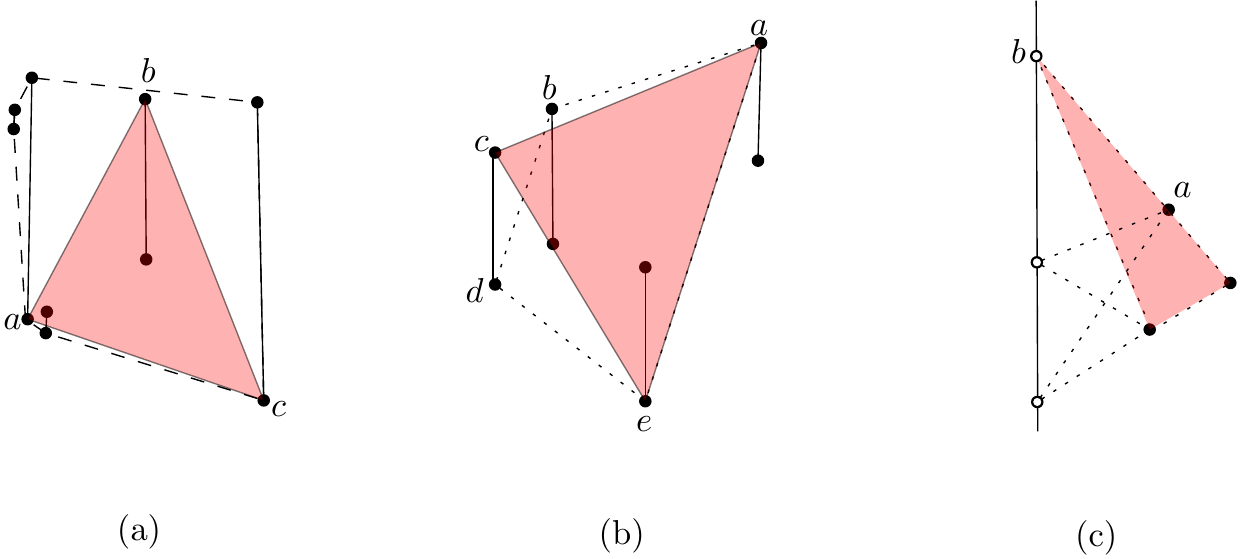}
	\caption{
		Three examples. In (a,b) the largest-area triangle is shown in red. The dashed polygon in (a) is the convex hull of the line segments; the dotted polygon in (b) is the largest-area convex hull possible by selecting a single point on each segment. 
		(a) The largest-area true triangle selects its vertices from the endpoints of the line segments, but not necessarily from those on the convex hull. (b) The largest-area true convex hull ($abde$) does not contain the maximum-area true triangle ($ace$). (c)  A set of one imprecise point and three fixed points, and the resulting largest-area strictly convex 4-gon. The inner angle at $a$ approaches $\pi$. Note that any convex 4-gon must include the fixed points, so the largest area occurs as the point on the segment approaches $b$.} 
	\label{fig0}
	
\end{figure}

\section{\mama problem}
 In this section, we will consider the following problem: given a set $L=\{l_1,...,l_n\}$ of parallel line segments, choose a set $P=\{p_1,...,p_n\}$ of points, where  $p_i \in l_i$,  such that the size of the largest-area triangle with corners at $P$ is as large as possible among all possible choices of $P$ (see Figure~\ref{mnmxmxmn}(a)). 
 Observe that this is, in fact, equivalent to finding three points on three different elements of $L$ that maximize the area of the resulting triangle.\footnote{Note that a similar statement will be true for \mimi, but not for \mami and \mima.} 
 First we  review several related previous results, then we discuss some difficulties that occur when  dealing with imprecise points.  
 
 Boyce \emph{et al.}~\cite{45} consider the problem of computing the largest-area  $k$-gon whose vertices are restricted to a given set of $n$ points in the plane, and prove that the optimal solution only uses points on the convex hull of the given point set (if there exist at least $k$ points on the convex hull).
 L\"offler and van Kreveld  \cite{39} proved that the maximum-area convex  polygon on  a given set of imprecise points (modeled as line segments) always selects its vertices from the convex hull of the input set.
 As a result, one might  lead to conjecture that the maximum-area triangle selects its vertices from the endpoints of regions on the convex hull. This is not the case, as can be seen in Figure~\ref{fig0}(a) (notice that the number of vertices of the convex hull is not fixed).
 Also, unlike in the precise context, the largest-area triangle is not inscribed in the largest possible convex hull of the given set of imprecise points, as illustrated in Figure~\ref{fig0}(b). This problem is further complicated for larger values of $k$, as illustrated in Figure~\ref{fig0}(c); even for $k$=4, we cannot find the area of the  maximum strictly convex $k$-gon, as the angle at $a$ approaches $\pi$ and we can enlarge the area of the convex 4-gon arbitrarily. We elaborate in Section~\ref{sec:generalk}.


\begin{obs} Let $L$ be a given set of imprecise points modeled as arbitrary length parallel line segments, and let $Z$ be its set of endpoints. If the largest-area triangle on $Z$ is not a true triangle in $L$, then the segment $l \in L$ that contributes two vertices to the largest-area triangle on $Z$, does not necessarily contribute a vertex to the largest-area true triangle on $L$. 
\end{obs}
\begin{proof}
The proof is done through providing a counter-example.
Consider a set of one imprecise point $bd$ and three  points $a$, $c$ and $e$, as illustrated in Figure~\ref{fig3}(a). The largest-area triangle is $abd$, but the largest-area true triangle is $ace$, where $e$ (resp. $c$) is located within the  (gray) apex of the wedge which is constructed by emanating two rays at $d$ (resp. $b $) parallel to $ac$ and $ab$ (resp. $ad$ and $ae$).
\end{proof}

\begin{obs} \label{obs:atendpoints}
	 Let $L$ be a given set of imprecise points modeled as parallel line segments. There is an optimal solution to the \mama problem, such that all the vertices are chosen at endpoints of the line segments.
\end{obs}

\begin{proof}
	 Suppose there exist three points $p$, $q$ and $r$ that form a true triangle, which has maximal area, and suppose that $p$ is not at an endpoint of $l_p$. If $p$ is not at an endpoint of its segment, we can consider a line $\ell$ through $p$ and parallel to ${qr}$. If we sweep $\ell$ away from ${qr}$, it will intersect $l_p$ until it leaves $l_p$ in a point $p'$. Clearly the triangle $p'qr$ has larger area than  $pqr$, and  $p'$ can be substituted for $p$ to give us a larger area true triangle,  and thus $pqr$ cannot be the largest-area  true triangle (if $\ell$ and $l_p$ are parallel, we can choose either endpoints and the area is the same).
	\end{proof}

\begin{obs} If at most two distinct line segments appear on $C_0$, there is an optimal solution to the \mama problem, such that  the two segments which appear on  $C_0$ always appear on the largest-area true triangle.
\end{obs}
\begin{proof}
From the previous observation we know that the  largest-area true triangle selects its vertices from the endpoints of the line segments.
Now we will prove that if at most two distinct segments appear  on $C_0$, these segments  always contribute to  the largest-area true triangle. 
Suppose this is not true. Then there are two cases.
First, suppose that there exist three points $p$, $q$ and $r$ at the endpoints of three different segments, such that $pqr$ has maximal area, and (w.l.o.g)  $p$ and $r$, respectively, have  the lowest and highest  $x$-coordinates among the vertices of $pqr$, and none of the vertices of $pqr$ are selected from the vertices of $C_0$, as illustrated in Figure~\ref{fig3}(b). We consider a line $\ell$  through $r$ and parallel to ${pq}$. If we sweep $\ell$ away from  ${pq}$, it will intersect $C_0$ until $\ell$ leaves $C_0$ at a point $r'$,  such that $r'$ should belong to one of the two segments that currently appeared on $C_0$. We also consider a line $\ell'$  through $p$ and parallel to ${qr}$. If we sweep $\ell'$ away from  ${qr}$, it will intersect $C_0$ until $\ell'$ leaves $C_0$ at a point $p'$, and  thus $p'$ and $r'$ can be substituted for $p$ and $r$ to give us a larger area true triangle, a contradiction.

Second suppose there exist three points $p$, $q$ and $r$ that form a true triangle, has maximal area, and  selects only  one of its vertices, e.g., $p$, from $C_0$.  Then $p$ either is the leftmost or the rightmost vertex of the triangle $pqr$. 
We will show that one of the other  vertices of  $pqr$ should also be on $C_0$. 
We consider a line $\ell$  through $r$ and  parallel to ${pq}$. If we sweep $\ell$ away from ${pq}$, it will intersect $C_0$ until $\ell$ leaves $C_0$ at a point $r'$, that  belongs to the other segment that currently  appears on $C_0$, and  thus  $r'$ can be substituted for  $r$ to give us a larger area true triangle. Contradiction.
\end{proof}

 From now on, we assume more than two distinct segments appear on $C_0$. Because otherwise we can easily solve the problem in $O(n)$ time; two distinct endpoints of the line segments appearing on $C_0$ determine the base of the largest-area true triangle.

 \begin{figure}
 	\centering
 	\includegraphics{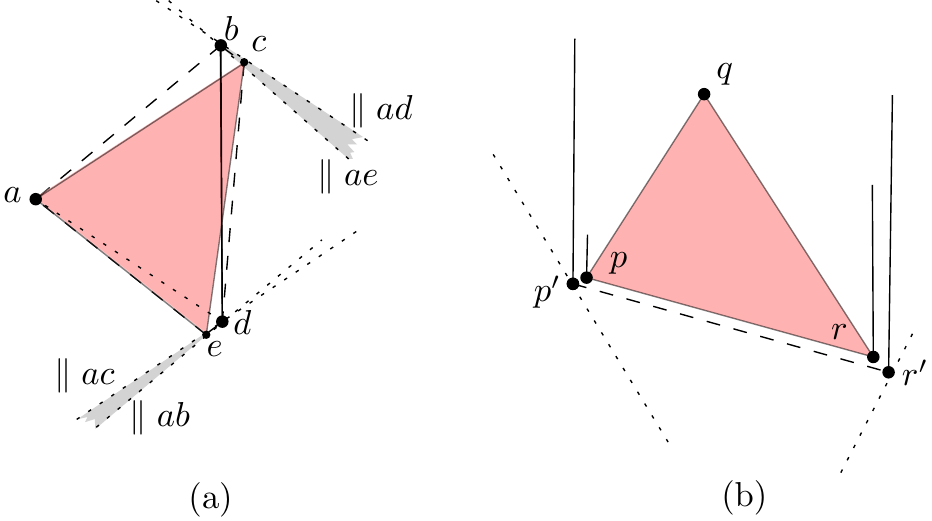}
 	\caption{
     (a) If the largest-area triangle is not a true triangle, then the segment that contributes two vertices to the largest-area triangle, does not necessarily contribute a vertex to the largest-area true triangle.	(b) If at most two distinct line segments appear on $C_0$, the two regions which appear on  $C_0$, always appear on the largest-area true triangle. 
}
 		\label{fig3}
 	\end{figure}

  \subsection{Equal-length parallel line segments} \label{sec:eqlen}

 We show that  for a set of equal-length parallel line  segments, the largest-area  triangle selects its vertices from the vertices on $C_0$ and is almost always a true triangle. 
  The only possible configuration of the line segments that makes the largest-area triangle a non-true triangle is collinearity of  all the upper (or lower) endpoints, in which case the largest-area true triangle will always select one vertex at the leftmost segment and one vertex at  the rightmost one.
 Clearly, we can test whether this is the case in $O(n)$  time.
In the following we show that if the upper endpoints are not collinear, we can directly apply any existing algorithm for computing the largest-area triangle on a point set.
 


 
\begin{lemma}
 Let $L$ be a set of equal-length parallel line  segments (not all the upper endpoints collinear). The largest-area true triangle selects its vertices from the vertices on $C_0$.
\end{lemma}

\begin{proof}
   In observation~\ref{obs:atendpoints} we  have proved that the vertices of largest-area true triangle are selected from the endpoints of their segments. 
 
Now we prove that the vertices of largest-area true triangle are selected from the vertices on  $C_0$. Suppose this  is false. Let $abc$ be the largest possible area true triangle, such that at least one of its vertices, e.g., $a$, is not located on $C_0$, as illustrated in Figure~\ref{prooftrue}(a). 
Then $a$ would be at the endpoint of  $l_a$, but it is not on the convex hull. 
Also  $l_a$ cannot  be the leftmost or the rightmost line segment, because otherwise at least one of the endpoints must appear on $C_0$. First suppose  $a$ does not have the lowest or highest  $x$-coordinate among the vertices of $abc$. But then either $l_a$ is located completely inside $C_0$, or  it has an endpoint $a''$  on $C_0$ (see Figure~\ref{prooftrue}(a,b)). 

We consider a line $\ell$ through $a$ and parallel to $bc$. If we sweep $\ell$ away from $bc$,  $\ell$ will intersect $C_0$ until it leaves $C_0$ in a vertex $q$.
If $q \notin \{l_b,l_c\}$, $q$ can be substituted for $a$ to give us a larger area true triangle,  a contradiction.

  Now suppose $q \in \{l_b,l_c\}$, as illustrated in Figure~\ref{prooftrue}(a,b). Let $b'$ and $c'$ denote the other endpoints of $l_b$ and $l_c$, respectively. But then there always exists another segment  $l_p$ which shares a vertex $p$ on $C_0$. Note that since the segments have equal length, $l_p \notin \{l_b, l_c\}$, (but $l_p$ can coincide with $l_a$).
  Thus   $p$ can be substituted for $a$, $b'$ can be substituted for $b$ and $c'$ can be substituted for $c$ to give us a larger area true triangle,  a contradiction. 


 Now suppose (w.l.o.g) $a$ has the lowest $x$-coordinate, as illustrated in Figure~\ref{prooftrue}(c). Consider the supporting line of $l_a$, $\ell$.  This crosses the upper and lower hull boundary at $a''$ and $a'$, respectively. Then $abc$ is interior to one half plane on $\ell$. Because of a convexity argument, and because  $a$ is strictly interior to the $C_0$,  there will be a point $p$ that lies in the opposite half plane  from the one which contains $b$ and $c$, such that either $a$ lies in the strip defined by  $bc$ and a line through $p$ parallel to $bc$, in which case $pbc>abc$, or $a$ does not lie in that strip and $a''bc>abc$. Both give a contradiction.
 
 
\end{proof}
Note that relations between triangles refer to their areas.

\begin{lemma}
\label{trutri}
 Let $L$ be a set of equal-length parallel line  segments. Unless the upper (and lower) endpoints of $L$ are collinear, the largest-area triangle is always a true triangle.
\end{lemma}
\begin{proof} 
	Suppose the lemma is false. Let $ l_l^{-} l_r^{-} l_l^{+}$ be the largest-area triangle on  a given set $L$ of imprecise points modeled as equal-length parallel line segments, and let ${l_l}$ and ${l_r}$ be the line segments that have the largest-area  triangle constructed on them, as illustrated in Figure~\ref{prooftrue}(d).   There cannot be any other line segments to the left  of ${l_l}$ and to the right  of ${l_r}$, because otherwise we can construct a larger area triangle by using  one of those line segments.
 Thus, all the other line segments must be located between ${l_l}$ and ${l_r}$. 
Suppose ${l_p}$ be one of them. It is easy to observe that $ l_l^{+} l_p^{-} l_r^{+}$ (or $ l_l^{-} l_p^{+} l_r^{-}$) will have a larger area than $ l_l^{-} l_r^{-} l_l^{+}$, because  with a fixed base length $ l_l^{-} l_r^{-}$, $ l_l^{+} l_p^{-} l_r^{+}$ has a  height longer than the length of the given input line segments, but   $ l_l^{-} l_r^{-} l_l^{+}$  has a  height smaller than the length of the input line segments. Contradiction.
\end{proof}

\begin{figure}
\includegraphics{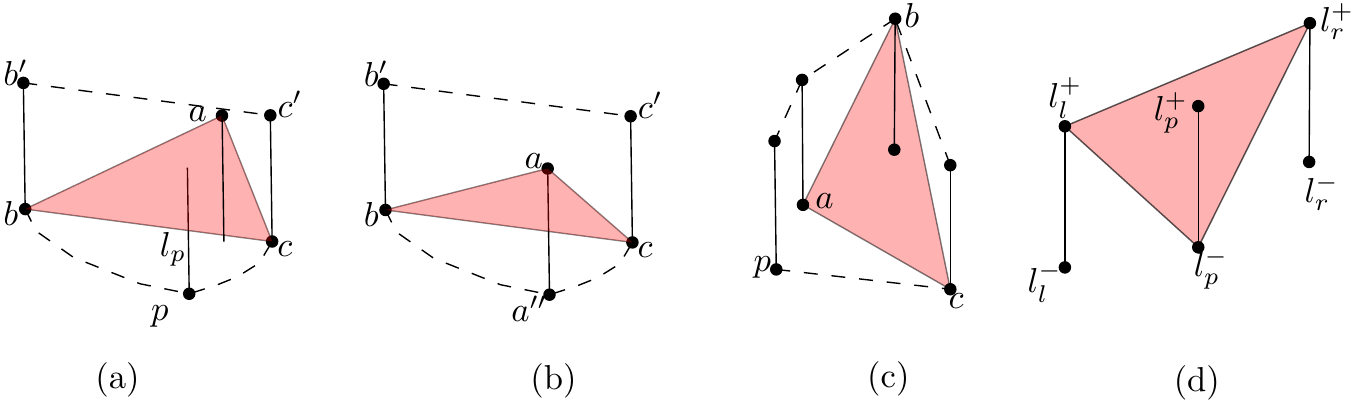}
\centering
\caption {
	(a,b,c) The largest-area true triangle on a given set of equal-length parallel line segments selects its vertices from the vertices on $C_0$.  (d) In a given set of equal-length  parallel line segments, the largest possible area triangle almost is always a true triangle.  }
\label{prooftrue}
\end{figure}

Recall that a rooted polygon is a  polygon  with one of its vertices fixed at a given point in a specific line segment.

\begin{theorem} Let $L$ be a set of $n$ imprecise points modeled as a set of disjoint parallel line segments with equal length. The solution of the problem \mama  can be found in $O(n \log n)$ time. 
\end{theorem}

\begin{proof} We first compute $C_0$.  From Lemma~\ref{trutri} we know  the largest-area triangle for a given set of equal-length parallel line segments is always a true triangle. Then,  we can apply the existing linear time algorithm to compute the largest-area inscribed triangle~\cite{chandran}. 
\end{proof}
\newpage
 \subsection{Arbitrary length parallel line segments} \label{sec:alen}

In this section, we consider the \mama problem for a given set of  arbitrary length parallel line segments. For simplicity we assume no two vertical line segments have the same
$x$-coordinates. 

 
 \begin{figure}
 	\includegraphics{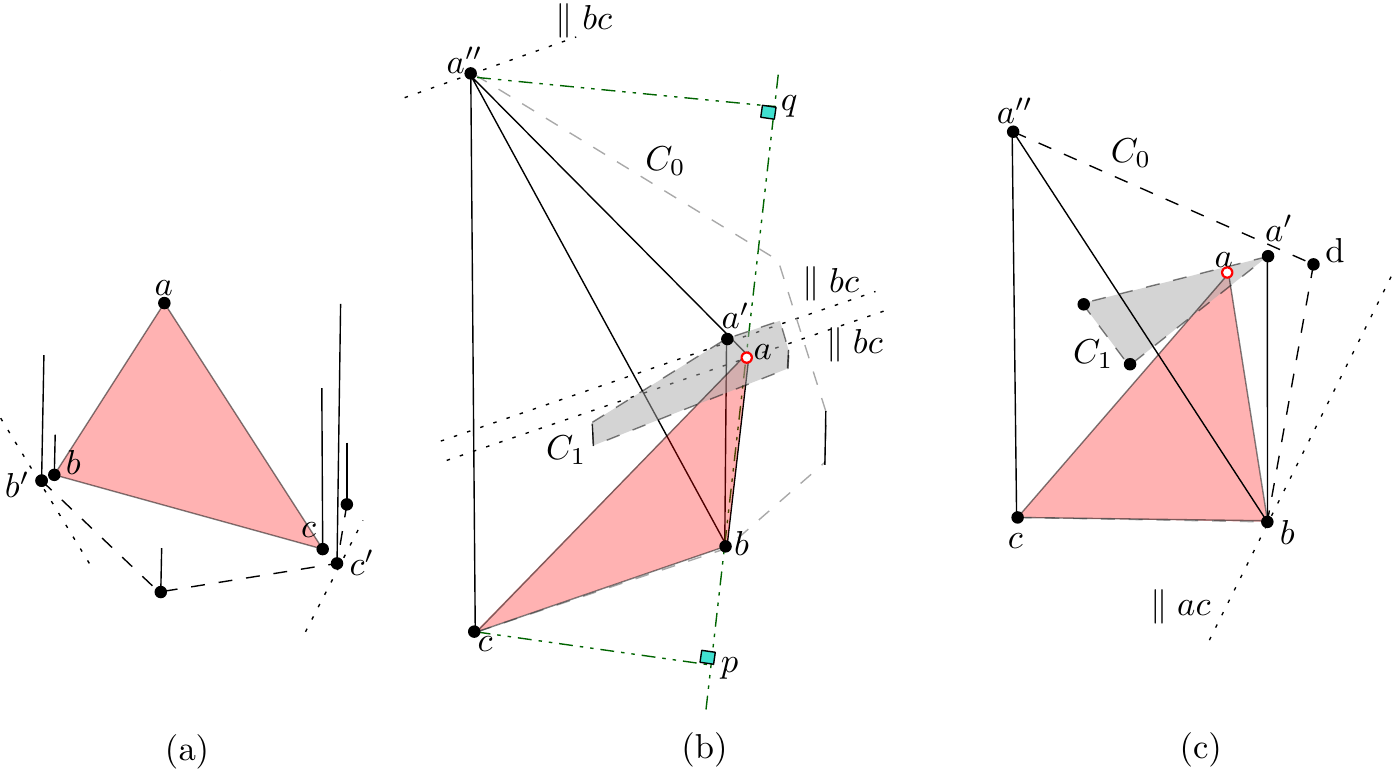}
 	\centering
 	\caption {(a)	The largest-area true triangle selects at least two vertices from the vertices of  $C_0$: if not, we can improve the area. (b,c) The largest-area true triangle cannot have a vertex inside $C_1$.}
 	\label{fig9}
 \end{figure}

\begin{lemma} \label{lem:chord}
At least two vertices of the largest-area true triangle are located on $C_0$, and at most one of its vertices is located on
 $C_1$, and all three are on $C_0 \cup C_1$. 
\end{lemma}
\begin{proof} First we will prove that at least two vertices of the largest area true triangle are located on $C_0$. Suppose this is not true,  so the largest-area true triangle  can have fewer vertices  on $C_0$. 

In the beginning, we  prove that it is not possible that none of the vertices of the largest-area true triangle are located on $C_0$. Let $ abc$ be the largest-area true triangle that does not select any of its vertices from $C_0$, and let $b$ and $c$ have  the lowest and highest $x$-coordinates, respectively, as illustrated in Figure~\ref{fig9}(a). We consider a line $\ell$  through $c$ (resp. $b$) and parallel to ${ab}$ (resp. $ac$). If we sweep $\ell$ away from ${ab}$ (resp. $ac$), it will intersect  $C_0$, until
  it leaves $C_0$ in a point $c'$ (resp. $b'$), and  $b'$ and $c'$ can be substituted for $b$ and $c$, respectively, to give us a larger area true triangle. Contradiction.
  
Again let $abc$ be the largest-area true triangle that  selects only  one of its vertices, e.g., $a$, on $C_0$.  We will show that at least one of the other vertices of  $abc$ must also be located on $C_0$.  
Suppose the $x$-coordinate of $a$ is between the $x$-coordinates of $b$ and $c$ (since otherwise with the same argument we had above, at least another vertex will also be located on $C_0$). We consider a line $\ell$  through $c$ and parallel to ${ab}$. If we sweep  $\ell$ away from ${ab}$, it will intersect $C_0$  until it leaves $C_0$ in  a point $c'$, that does not  belong to the regions of $l_a$ and $l_b$. We can also do the same procedure for the point $b$ and find another point $b'$, and thus, $b'$ and $c'$ can be substituted for $b$ and $c$, respectively, to give us a larger area true triangle.  Contradiction. Thus, at least two vertices of the largest-area true triangle should be located on $C_0$.  

Second we will prove that if the third vertex of the largest-area true triangle is not located on $C_0$, it must be located on $C_1$. Suppose this is false.

 Let $ abc$ be the largest-area true triangle, and (w.l.o.g) let $a$ be the vertex that is not located on $C_0$.  
 Suppose this is not true, that $a$ cannot be located interior to   $C_1$. Notice that $b$ and $c$ must be located on  $C_0$. We consider a line $\ell$ through $a$ and parallel to ${bc}$. If we sweep $\ell$ away from  ${bc}$,  it will intersect $C_1$ until it leaves $C_1$ in a point $a'$. If this point does not belong to $l_b$ or $l_c$, we are done.
 Suppose $a' \in {l_b,l_c}$.  We continue sweeping $\ell$ until it leaves $C_0$ in a vertex  $a''$. Suppose again that it belongs   to one of $l_b$ or $l_c$ (since otherwise we have found a larger area true triangle and we are done).  Note that we also did not find  any other vertex during the sweeping $\ell$ away from $bc$. 

 Suppose that $a''q$ and $cp$ are the perpendicular segments to the supporting line of $ab$ from $a''$ and $c$, respectively, as illustrated in Figure~\ref{fig9}(b). 
First suppose  both of the $l_b$ and $l_c$ are located to the  left (or right) of $l_a$. Since $c$ and $a''$ belong to the same segment, and from  the slope of the supporting line of  ${ba}$ we  understand that the intersection of the supporting lines of $a''c$ and ${ba}$ would be to the right  of the supporting line of $\overrightarrow{cb}$. It follows that the triangle $ aa''b>abc$, since with a fixed base ${ab}$, the height ${a''q}$ is longer than ${cp}$. Contradiction. 

Now suppose only  one of the line segments $l_b$ or $l_c$ is located to the  left (or right) of $l_a$. Without loss of generality, suppose $l_c$ is located to the left of  $l_a$, as illustrated in Figure~\ref{fig9}(c). 
Let $d$ be the next vertex of  $a''$  on the cyclic ordering of $C_0$ (note that $d$ always exists since $C_0$ and $C_1$ are disjoint). Obviously $ a''db$ is always a true triangle. 
 Since  $ a'a''b =  a'cb$, $ a'cb >  acb$ and  $ a''bd >  a''ba'$, we would have  $ a''db >  abc$. Contradiction.   
\end{proof}

 \begin{corollary} \label{cor:correctness}
  Let $C_0=\{p_1,\ldots,p_m\}$ be a  convex polygon,  constructed at the endpoints of a set of arbitrary-length parallel line segments. There exists  an $i$ ($1 \leq i \leq m$) such that  the largest-area   true triangle rooted at $p_i$ is the largest-area true triangle inscribed in $C_0$.
 \end{corollary}

We start solving the \mama problem by considering the case where all the vertices of the largest-area true triangle are the vertices on $C_0$. The case where one of the vertices of the optimal solution is located on $C_1$ will be considered later. 

\subsubsection{Algorithm: Largest-area true triangle inscribed in $C_0$}
We first compute $C_0=\{p_1,...,p_m\}$, in a representation where vertices are ordered along its boundary in counterclockwise (ccw) direction. The idea is to find  the largest-area true triangle with base on all the chords of $C_0$. We start our algorithm from an arbitrary vertex  $p_1$  as the  root. We  check which of the possible true chords  $p_1p_i$, where $i=2,...,m$, will form a larger-area true triangle.  Let $p_j$ be the furthest vertex (in vertical distance) from true chord  $p_1p_i$, then either $p_j$, or one of the two previous or next neighbors of  $p_j$ 
will construct the largest-area true triangle in the base $p_1p_i$.
If  $p_1p_ip_j$ is not true,   at most four triangles   $p_1p_ip_{j-1}$, $p_1p_ip_{j+1}$, $p_1p_ip_{j-2}$ and  $p_1p_ip_{j+2}$ are candidates of constructing the largest-area true triangle on the base $p_1p_i$. Note that a chord always divides a convex polygon into two convex polygons.  We look for the largest-area true triangle on the base $p_1p_i$ on both halfs of $C_0$ simultaneously.

In the next step of the algorithm, we start looking for the furthest  vertex of $p_1p_{i+1}$ (if it is a  true chord) linearly from $p_j$ and  in counterclockwise direction. We  stop looking on $p_1$ when we reach to the chord $p_1p_{m}$. 

We repeat the above procedure in  counterclockwise order, where in each step, we consider one vertex of $C_0$ as the root. We stop the algorithm when we return to $p_1$, and report the largest-area true triangle we have found.
This algorithm is outlined in Algorithm~\ref{alg:impquad}  and illustrated in Figure~\ref{fig:quadalg}. It is easy to observe that  Algorithm~\ref{alg:impquad} takes $O(n^2)$ time.

Now we give a lower bound for computing the largest-area triangle rooted at a given vertex $r$. 
 \begin{figure}
	\includegraphics{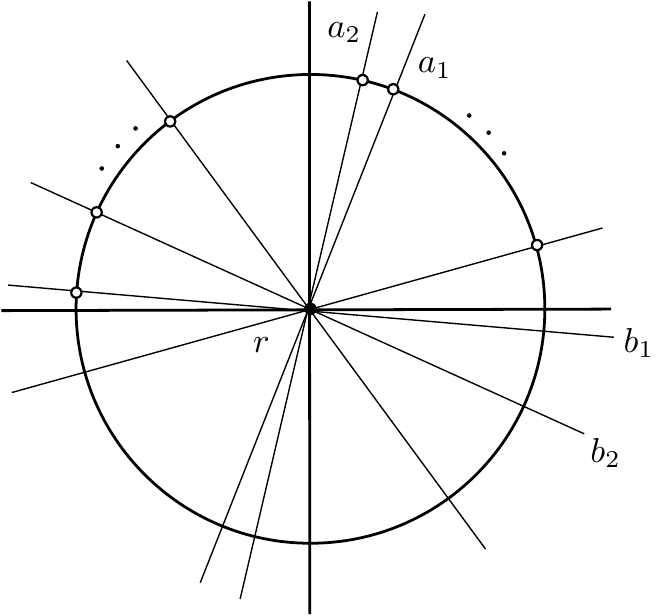}
	\centering
	\caption {Illustration of the reduction from set disjointness problem.}
	\label{fig:lb}
\end{figure}

\begin{theorem}
	There exists a lower bound $\Omega (n \log n)$ for computing the largest-area true triangle rooted at a given vertex $r$.
\end{theorem}

\begin{proof}
Our reduction follows
the set disjointness problem which has an $\Omega(n\log n)$ lower bound in the algebraic
decision tree model~\cite{preparata1985computational}: 
Given two sets  $A=\{a_1, a_2, . . . , a_n\}$, $B=\{b_1, b_2, . . . , b_n\}$, determine whether or not $A\cap B =\emptyset$. 

We map each $a_i$ to the line $y = a_ix$. 
Also each $b_i$ is mapped to the line $y = -1/b_ix$.
Consider the set of $2n$ intersection points of these lines with the first and second
quadrants of the unit circle, centered at the origin $r$ of the coordinate system, as illustrated in Figure~\ref{fig:lb}. 
Clearly, the maximum size of the largest-area triangle rooted at $r$ can be $1/2$. But a  triangle of this size appears   
if and only if  there exists   two elements $a_i \in A$ and $b_j \in B$, where $a_i=b_j$.\footnote{We assume w.l.o.g that $A$ and $B$ do not include $0$.} 
Thus there exist two elements  $a_i \in A$ and $b_j \in B$, where $a_i=b_j$, if and only if either the size of the largest-area triangle rooted at $r$ equals $1/2$.
\end{proof}

Note that our algorithm runs in  $O(n^2)$ time, while   $O(n)$  points are considered as root $r$.\footnote{Reducing this running time is an interesting open problem, even when the input is a set of points (see, e.g.,~\cite{mukhopadhyay2013all}) .} Also notice that this lower bound is only for the case where we do not know the sorted order of the points on $C_0$, and the line segments are very short.

\begin{algorithm}[H] 
\caption{ Largest-area true triangle inscribed in $C_0$}
	\label {alg:impquad}
		{\bf Legend} Operation {\texttt{\textit{next}} means the next vertex in ccw order}\\
	 {\bf Procedure} {\sc Largest-Inscribed-Triangle($C_0$)} \\
	{\bf Input} {$C_0={p_1,\ldots,p_m}$: convex polygon of the segments, $p_1$: a vertex of $C_0$}\\
	{\bf Output} {$T_{max}$: Largest-area true triangle}\\
	$a$ = $p_1$\\
	$max$ = 0\\
		\While{true}
		{
$b$ = \texttt{next}($a$)\\
	\If{b $\in l_a$}
	{
	$b$ = \texttt{next}($b$)\\
	}
 		$c_l,c_r$ = farthest  vertices from $ab$ (in vertical distance and on both halfs of $C_0$)\\
		\While{$c_l \ne a$}
		{
		\If{$c_l$ (resp. $c_r$) $\in \{l_a,l_b\}$ }
		{
		update $c_l$ (resp. $c_r$)  with another vertex with farthest  distance from $ab$ (among two previous or next neighbors of $c_l$),  such that $abc_l$ (resp. $abc_r$)  is true.
    	       }		
	
			${max}$ = $max(abc, abc_l,abc_r)$
		    
		$b$ = \texttt{next}($b$)\\
		\If{b $\in l_a$}
		{
		$b$ = \texttt{next}($b$)\\
		}			
		\While{$c_l$ (resp. $c_r$ ) is not the farthest  vertex from $ab$}
			{
				$c_l$ = \texttt{next}($c_l$)  (resp. $c_r$ = \texttt{next}($c_r$)) \\
			}
			
		}	  
	$a$ = \texttt{next}($a$) \\
	\If{$a$=$p_1$}
	{
	\Return ${max}$\\
	} 
				
}
	
\end{algorithm}

\begin{figure}
	\includegraphics{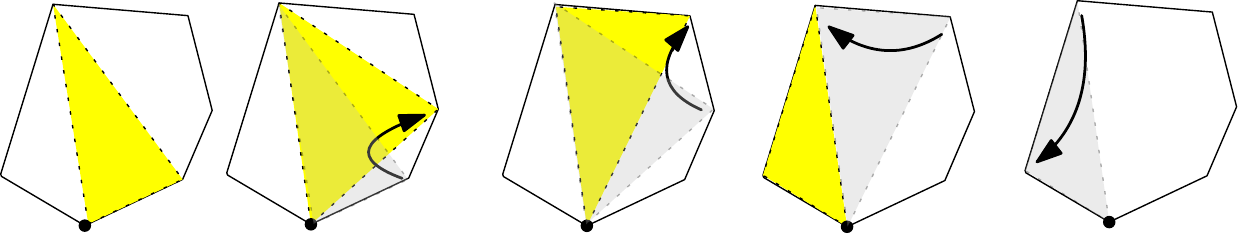}
	\centering
	\hspace {1cm}
	\caption {The first  step of Algorithm~\ref{alg:impquad}}
	\label{fig:quadalg}
\end{figure}


\begin{lemma} \label{ltrueinc0}
Algorithm~\ref{alg:impquad}  finds the largest-area true triangle inscribed in $C_0$.
\end{lemma}

\begin{proof}
 
 From Lemma~\ref{lem:chord} and Corollary~\ref{cor:correctness} we know  the  largest-area true triangle inscribed on  $C_0$ will be constructed on a true chord of $C_0$. In Algorithm~\ref{alg:impquad} we  consider  all the vertices of $C_0$  as the root, and we compute the largest-area true triangle on each true chord of $C_0$. Thus the algorithm works correctly.
\end{proof}

\begin{corollary}
\label{3c}
 Let $L$ be a set of imprecise points modeled as a set of $n$ parallel line segments with arbitrary length. The largest-area true triangle which selects its vertices from the vertices on $C_0$ can be found in $O(n^2)$ time. 
\end{corollary}
 

\subsubsection{Algorithm: Largest-area true triangle }
From Lemma~\ref{lem:chord} we know the combinatorial structure of the largest-area true triangle:  it can be the largest-area true triangle inscribed in $C_0$, or the largest-area true triangle that selects two neighboring vertices on $C_0$ and one vertex on $C_1$, or the one that selects two non-neighboring vertices on $C_0$ and one vertex on $C_1$. The  largest-area true triangle is the largest-area triangle among them.

In the first case, from  Corollary~\ref{3c}  the largest-area true triangle can be found in  $O(n^2)$ time.  

In the second case, we consider  all the true edges of $C_0$  as the base of a  triangle. The third vertex of each triangle can be found by   a binary search  on the boundary of $C_1$.
 Let $bc$ be a true edge of $C_0$. If  a triangle $abc$  on the base  $bc$   is not true, one of the   next or previous neighbors of    $a$  on  the cyclic ordering of $C_1$ may construct the largest-area true triangle on $bc$ (see Figure~\ref{vertex4}(a,b)). In fact if  the next or previous neighbor of $a$ also belongs to $l_b$ or $l_c$, then  there always exists a vertex $a'$ on $C_0$ that makes a larger  true triangle on $bc$,  as illustrated in Figure~\ref{vertex4}(b).
  Thus, in this case again, the largest-area true triangle can be computed in $O(n \log n)$ time.

In the third case, there are $O(n^2)$ chords to be considered as the base of the largest-area true triangle. Since a chord of $C_0$ may decompose $C_1$  into two convex polygons, we should do a binary search   on each half of  $C_1$, separately. 
 This method costs $O(n^2 \log n)$ time totally. But we  can still  do better.

Let $r$ be any vertex on $C_1$; we will search for the largest-area rooted triangle on $r$. An axis system centered on $r$  will partition $C_0$  into four convex chains, so that the  largest-area true triangle should be rooted at $r$ and two points on the  other quadrants, as illustrated in Figure~\ref{vertex4}(c).

From Corollary~\ref{3c} we know if one quadrant, or two consecutive quadrants include the  other  vertices of the largest-area true triangle, we can find the largest-area true triangle in $O(n^2)$ time, since it will be constructed on the vertices of the boundary of a convex polygon. 

Suppose  two other vertices belong to two diagonal quadrants, e.g., quadrant one and quadrant three, as illustrated in Figure~\ref{vertex4}(d).  Let the cyclic ordering of $C_0$ be counterclockwise, and let $v_1$ and $v_2$ be the first and second vertices of the third quadrant in the cyclic ordering of $C_0$. 
Then we can find $f_1(v_1)$ and $f_2(v_1)$ on quadrant one, so that $v_1$,  $r$ and each of $f_1(v_1)$ and $f_2(v_1)$ construct the largest-area true triangle on the base ${v_1r}$, as we can see in Figure~\ref{vertex4}(d).
 We named $f_1(v_1)$ and $f_2(v_1)$ the \emph{starting points}. The starting points determine the starting position of looking (in ccw direction) for $f_1(v_2)$ and $f_2(v_2)$, so that  $v_2$,  $r$ and each of $f_1(v_2)$ and $f_2(v_2)$ construct the largest-area true triangle on the base ${v_2r}$, etc.

Thus, for any fixed $r \in C_1$ and any vertex $v_i \in C_0$ on the third  quadrant, we start looking for $f_1(v_{i})$ and $f_2(v_{i})$  from $f_1(v_{i-1})$ and $f_2(v_{i-2})$ (in ccw direction) on the first quadrant, respectively.

\begin{figure}
	\includegraphics{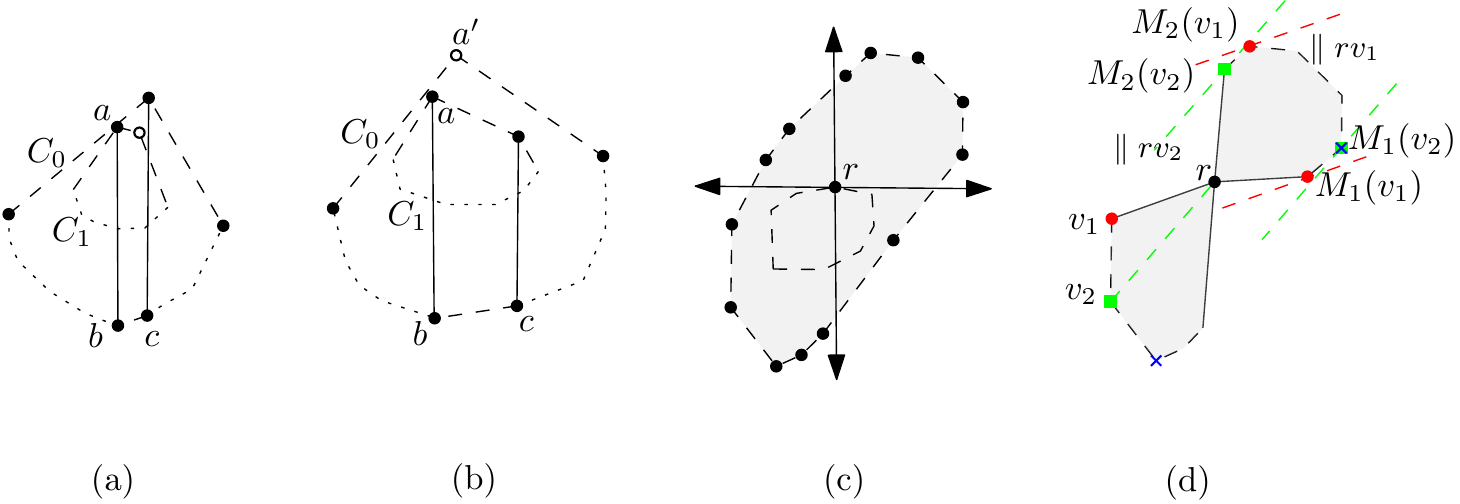}
	\centering
	\caption { (a,b) Possible cases of observing a   non-true triangle $abc$ during  a binary search on the boundary of  $C_1$. (b) If the next or previous neighbor of $a$  belongs to $l_b$ or $l_c$, there always exists a vertex $a' \in C_0$ which makes  a larger true triangle on the base $bc$. 
		(c) Selection of point $r$ on $C_1$ as the origin. (d)  In the diagonal quadrants, for a  vertex $v_2$ in the third quadrant, we only need to look  for the third vertex of  the largest-area true triangle (on the base $v_2r$) in the first quadrant   from $f_1(v_{1})$ and $f_2(v_{1})$.}
	\label{vertex4}
\end{figure}
Note that for any  vertex $r$ on $C_1$, the rooted triangle at $r$ can also be a non-true triangle.  Suppose we are looking for $f_1(v_{1})$ and $f_2(v_{1})$ of  vertex $v_1$ in quadrant three. If $r$ belongs to $l_{v_1}$, we discard $v_1$. 
 Let $r \in l_{f_1(v_1)}$ (or $r \in l_{f_2(v_1)}$). Then one of the next or previous neighbors of $f_1(v_1)$ should be the area-maximizing vertex on the base $v_1r$, and this point cannot belong to $l_{v_1}$ or $l_r$.

 Therefore, in the case where we are looking on diagonal quadrants, e.g., quadrant one and quadrant three, for any vertex $r \in C_1$, we first find the  starting points for the first vertex in quadrant three in $O( \log n)$ time, and then we  only  wrap around $C_0$ in at most one direction and never come back. Thus, for any $r \in C_1$,  the largest-area true triangle can be found in $O(n)$ time, and the largest-area true  triangle  can be found in $O(n^2)$ time totally. The whole  procedure of computing the largest-area true triangle is outlined in Algorithm~\ref{alg:mama}.
 \begin{obs}
 	Let $L$ be a set of $n$ imprecise points modeled as  parallel line segments with arbitrary length. The  largest-area rooted true triangle can be found in $O(n \log n)$ time.
 \end{obs}
 
 \begin{proof}
 	Let $r$ be the root. In Corollary~\ref{cor:correctness} we considered  the case where $r$ is a vertex on the boundary of $C_0$.
 	If $r$ is a point inside $C_0$, we set $r$ as the origin and compute the   four possible quadrants of $C_0$. It is easy to observe that the largest-area  true triangle rooted at $r$ can be found in $O(n \log n)$ time. 
 	
 \end{proof}
 
 \begin{algorithm}[H]
 	\caption{\mama  }
 	\label {alg:mama}
 	{\bf Legend} Operation {\texttt{next} means the next vertex in counterclockwise order}\\
 	{\bf Input} {$L=\{l_1,\ldots,l_{n}\}$}\\
 		
 	{\bf Output} {: largest-area true triangle}\\
 	$Z$= the set of all the endpoints of elements of $L$\\
 	$C_0$=  $CH(Z)$\\
 	$C_1$= $CH(Z \setminus \partial CH(C_0))$\\
 	$T_{C_0}$= {\sc Largest-Inscribed-Triangle($C_0$)}\\
 		\While{there is an unprocessed vertex $r$ on $C_1$}
 		 	{
 		 		$Q_1,Q_2,Q_3,Q_4$=Partitionioning $C_0$ into four convex chains (in ccw direction) by considering an axis system centered on $r$.\\
 		 		$T_{C_0,C_1}$=Max(Largest-Inscribed-Triangle($Q_1 \cup Q_2 \cup \{r\}$),
 		 		{\sc Largest-Inscribed-Triangle($Q_2 \cup Q_3 \cup \{r\}$)},
 		 		{\sc Largest-Inscribed-Triangle($Q_3 \cup Q_4 \cup \{r\}$)},
 		 			{\sc Largest-Inscribed-Triangle($Q_4 \cup Q_1 \cup \{r\}$)})\\
 		 			$v$= the first vertex of $Q_3$ \\
 		 			\While{there is an unprocessed vertex $v$ $\in$ $Q_3$}
 		 			{
 		 				\If{v $\in l_r$}
 		 				{
 		 					v = \texttt{next}(v)\\
 		 				}	
 		 				$f_1(v),f_2(v)$= the farthest vertices from $vr$ on $Q_1$\\
 		 				
 		 				\If{$f_1(v)$ or $f_2(v)$ $\in \{l_r,l_v\}$}
 		 				{
 		 					update $f_1(v)$ or $f_2(v)$ with its next or previous neighbor which has the farthest vertical distance from $vr$\\
 		 				}
 	 			$T_{C_0,C_1}$=Max($vrf_1(v)$,$vrf_2(v)$,	$T_{C_0,C_1}$)	\\
 		 		$v$=\texttt{next}($v$)\\
 		 	}
 		 		repeat above while loop for $Q_2$ and $Q_4$\\
 		 		$r$= \texttt{next}($r$)
}
 	
 	\Return\textsc Max($T_{C_0}$,$T_{C_0,C_1}$)\\
 	
 \end{algorithm}
 %

\begin{figure}
	\includegraphics{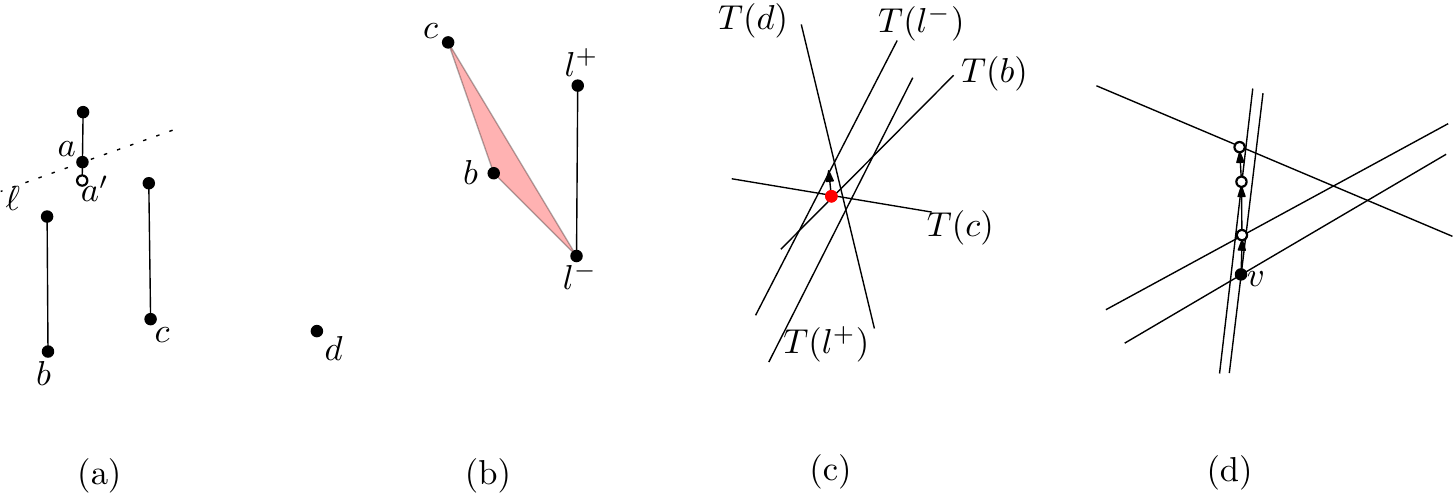}
	\centering
	\caption {
		(a) The smallest-area true triangle selects its vertices at  the endpoints of line segments. (b,c) The smallest-area true triangle on a set consists of one imprecise and three
		single points, and the optimal solution  in the dual space. (d) We  need to continue the movements in the up and down direction,  when we encounter a line   which is parallel to the ones intersecting at $v$.}
	\label{minar}
\end{figure}

\begin{theorem}\label{mamaquad}
 Let $L$ be a set of $n$ imprecise points modeled as a set of parallel line segments with arbitrary length. The solution of the problem \mama  can be found in $O(n^2)$ time. 
\end{theorem}

\newpage
\section{\mimi problem} \label{sec:minmin}
 In this section, we will consider the following problem: given a set $L=\{l_1,...,l_n\}$ of parallel line segments, choose a set $P=\{p_1,...,p_n\}$ of points, where  $p_i \in l_i$,  such that the size of the smallest-area triangle with corners at $P$ is as small as possible among all possible choices of $P$ (see Figure~\ref{mnmxmxmn}(b)).
As in the case of \mama, this problem is equivalent to finding three points on distinct elements of $L$ that minimize the area of the resulting triangle.

 The problem of finding the smallest-area triangle in a set of $n$ precise points is 3SUM-hard\footnote{The class of problems which  is unknown to be solvable in $O (n^{2-\epsilon})$ time for some $\epsilon>0$.}, as it requires testing whether any triple of points is collinear.
   In our case, if we find three collinear points on three distinct segments, the smallest-area true triangle would have zero area. 
   \begin{obs}
   	Let $L$ be a set of $n$ vertical line segments. Deciding whether there are three collinear points on three distinct line segments can be done in $O(n^2)$ time.
   \end{obs}
\begin{proof}
	The idea is look at the vertical segments in the dual space, where each vertical segment is transformed to a strip. Then, if there exists a common point in the strips corresponding to  three distinct line segments, this point denotes a line  passing through those segments in the primal space. This can easily be checked by a topological sweep of the arrangement of strips in which a curve $l$ sweeps over the intersection points, e.g., from left to right in the dual space, while $l$ keeps track of the required information about the intersected strips by $l$. Since a topological sweep of an arrangement can be done in $O(n^2)$ time~\cite{EDELSBRUNNER1989165}, we can test whether  there are three collinear points on three distinct segments  in $O(n^2)$ time.
	\end{proof}
   
 In the following we will show that when  the  smallest-area true triangle has non-zero area, it  selects its vertices on the endpoints of the line segments. 
 
 \begin{lemma}
\label{mina}
  Let $L$ be a set of imprecise points modeled as a set of parallel line segments.  Suppose there is no zero-area triangle in $L$. The smallest-area true triangle selects its vertices on the endpoints of the line segments.
\end{lemma}

\begin{proof} Suppose the lemma is false. Let $abc$ be the smallest-area true triangle, and suppose at least one of its vertices e.g., $a$, is not located at the endpoints  of $l_a$ (see Figure~\ref{minar}(a)). Consider a line $\ell$ through $a$ and parallel to $bc$. If we sweep $\ell$ towards $bc$, it will  intersect $l_a$ until it leaves   $l_a$ at a point $a'$.  Thus $a'$ can be substituted for $a$ to give us a smaller area true triangle, a contradiction.
\end{proof}

 We now introduce some notation. Let $L$ be a set of imprecise points modeled as parallel line segments, and let $Z$ be the set of all endpoints of $L$. For a given point $p=(p_x,p_y)$, we consider $T(p)=p_x x-p_y$ as  a transformation of $p$ to a dual space, and $A(Z)$ as the arrangement of the lines that are the transformation of $Z$ in dual space, as  Edelsbrunner \emph{et al}.~\cite{Edel} defined for a given set of points in the plane. 
  Let $T(l_a^{+})$ (resp. $T(l_a^{-})$) denote the transformation of the upper (resp. lower) endpoint of $l_a$ in the dual space.
 $A(Z)$ partitions the plane into a set of convex regions, and  $A(Z)$ has  total complexity of $O(n^2)$.  
 
  \begin{figure}
 	\includegraphics{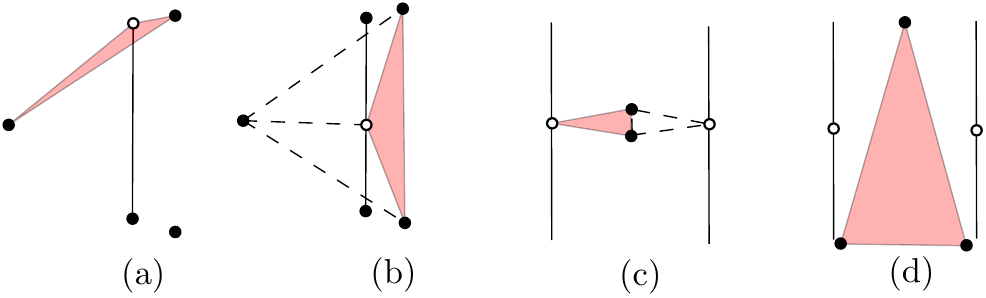}
 	\centering
 	\caption {Some observations on the \mami and the \mima problems. (a) The solution of the  largest smallest-area triangle 
 		on a set consisting of one imprecise point and three
 		single points, where it must be located at the endpoints.  (b) The  largest smallest-area triangle of (a). (c) The solution of smallest largest-area triangle on a set consisting of two imprecise points and two fixed points selects the interior points of the line segments. (d) The smallest largest-area triangle on a set consisting of two imprecise points and three fixed points.}
 	\label{notendp}
 \end{figure}

\subsection{Algorithm}
From  Lemma~\ref{mina} we know that we  only need to consider the endpoints of the line segments. For a given set $L$ of imprecise points with the endpoints in $Z$, we first construct $A(Z)$. Let $abc$ be the smallest-area triangle in the primal space. Each vertex $ v=\{T(l_a^{+,-}) \bigcap T(l_b^{+,-})\}$ in a face  $F$ in $A(Z)$ belongs to two different segments, since the endpoints of a line segment are mapped to two parallel lines in the dual space, (see Figure~\ref{minar}(c)). 
For  $v \in F$,  we first consider the true edges of $F$  (which belong to  distinct segments in the primal space).
 If edge $e \subset T(l_c^{+,-})$ is the first candidate for constructing the smallest-area  triangle on $v$,  it should be immediately located above or below $v$ among all the 
lines, as our duality preserves the vertical distances.
 If  $e$ does not determine a distinct segment in the primal space,  we should continue our movement  in both up and down directions until we reach a line in dual space which determines  a distinct line segment in primal space. In this situation, we may need to cross among two neighboring faces of $F$, as illustrated in Figure~\ref{minar}(d).  Note that since we can use our procedure to determine if 3  points of $n$ distinct line segments in the primal space are collinear, this problem is also 3SUM-hard.

 Thus we can compute the  smallest-area true triangle that can be constructed on each vertex $v \in A(Z)$. If it is smaller than previously computed optimal solution, $T_{min}$, we remember it. If we store $A(Z)$ in a reasonable data structure (e.g. doubly connected edge list),  the time required to find the answer is $O(n^2)$~{\cite[Corollary 2.5]{Edel}}.

\begin{lemma}
Let  $L$  be a set of imprecise points modeled as parallel line segments, and let $A(Z)$ be the arrangement of the transformation of $Z$ (which consists of the endpoints of $L$). There exists a face $F$ in $A(Z)$ such that the smallest-area true triangle uses one of its vertices.
\end{lemma}

\begin{proof}
The correctness proof of this lemma comes from these facts: all the vertices of $A(Z)$ belong to distinct segments in primal space,  our duality preserves the vertical distances and it is order preserving, and  we will find the smallest-area true triangle on each vertex of $A(Z)$, even when we encounter to non-distinct segments.  
\end{proof}

\begin{theorem}
Let  $L$  be a set of $n$ imprecise points modeled as parallel line segments. The solution of the problem \mimi can
be found in $O(n^2)$ time.
\end{theorem}

\newpage
\begin{figure}
	\includegraphics{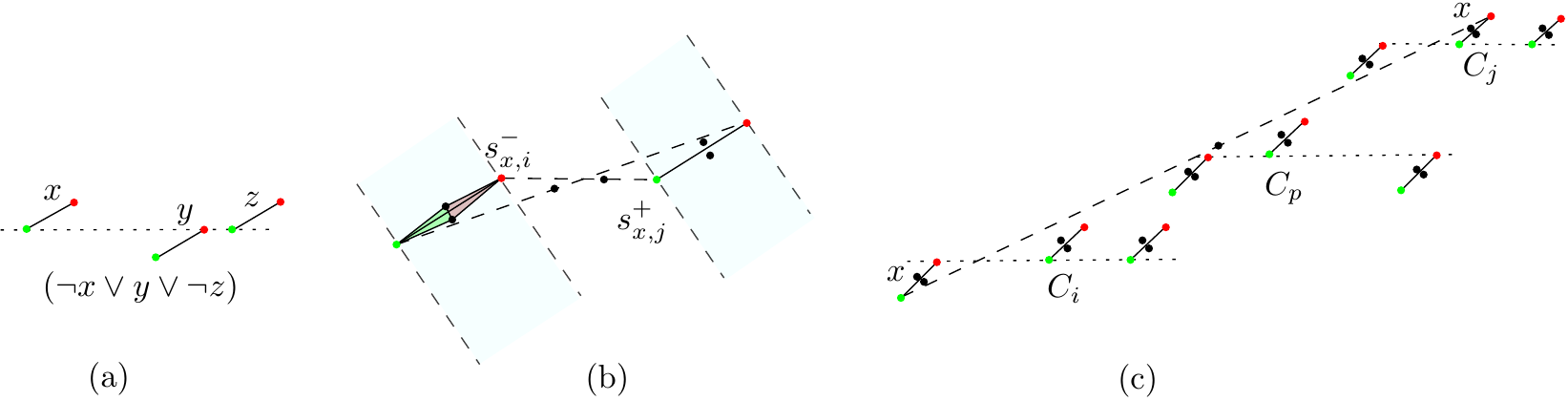}
	\centering
	\caption{(a) A clause gadget. (b) The structure of a common variable $x$ between two clauses $C_i$ and $C_j$.  (c) Three clause gadgets in a SAT instance.}
	\label{clause}
\end{figure}
\section{\mami problem} \label{sec:maxmin}
The problem we study in this section is the following: given a set $L=\{l_1,...,l_n\}$ of parallel line segments, choose a set $P=\{p_1,...,p_n\}$ of points, where  $p_i \in l_i$,  such that the size of the smallest-area triangle with corners at $P$ is as large as possible among all possible choices of $P$ (see Figure~\ref{mnmxmxmn}(c)).
 As we can see in  Figure~\ref{notendp}(a,b), the solution does not necessarily select its vertices on the endpoints of the line  segments.  We show this problem is hard.

\subsection{ \mami problem is NP-hard}
  We  reduce from SAT. 
  Given a SAT instance, we choose a value $\alpha$ and create a set of line segments such that if the SAT instance is satisfiable, there exists a point placement with a smallest-area triangle of area $\alpha$, but if the SAT instance is not satisfiable, every possible point placement will admit a triangle of area smaller than $\alpha$.
  
  We define a  variable gadget for each of the variables in a clause, and a clause gadget, that includes all of its variable gadgets.
A variable gadget  for a variable $x$ in a clause $C_i$ consists of a segment $s_{x,i}$  and two fixed points (degenerate line segments) on the bisector and close to  the center of  $s_{x,i}$ (see Figure~\ref{clause}).\footnote{For ease of construction, we draw the segments diagonal and the clauses horizontal, a simple rotation yields a construction for vertical segments.} These fixed points determine  two small triangles (the green and purple triangles) of area $\alpha$, and one of them  must be part of any point placement. 
We choose $\alpha$ small enough such that in the remainder of the construction, every other possible triangle either has an area larger than $\alpha$, or a zero area.



 Let  $s^+_{x,i}$ represent the value \textit{True} and $s^-_{x,i}$ represent the value \textit{False}, where $s^+_{x,i}$ and $s^-_{x,i}$ are endpoints of $s_{x,i}$.  
 We place these endpoints in such a way that if a clause is not satisfied, the corresponding endpoints will form a triangle with zero  area. For example, if the clause $C_i$ is $(\neg  x \vee y\vee \neg  z)$, the endpoints $s^+_{x,i}$, $s^-_{y,i}$ and $s^+_{z,i}$ will be collinear.
 So, setting $x$ and $z$ to \textit {True} and $y$ to \textit {False} will result in  a zero area triangle, as we can  see in Figure~\ref{clause}(a). 
In order to ensure segments representing the same variable will be assigned the same value, we place a fixed point on every line through the \textit{True} end of one and the \textit{False} end of another such segment.
For instance, if variable $x$ occurs in $C_i$ and $C_j$, we place a fixed point on the segments connecting $s^+_{x,i}$ to ${s^-_{x,j}}$, and one on the segment connecting $s^-_{x,i}$ to ${s^+_{x,j}}$ (see Figure~\ref{clause}(b)).  If $x$ selects both of the values of \textit{True} and \textit{False} simultaneously,  it will construct a zero area triangle. We will construct such structures for all the common variables of the clauses, see Figure~\ref{clause}(c). 


We must take care when placing all the endpoints of the line segments in the variable gadgets and the fixed points, that there are no  three collinear points other than those that are collinear by design. 
We must also make sure to not place any fixed points or  endpoints of segments  in the perpendicular strips which are determined by each line segment in the variable gadgets (see Figure~\ref{clause}(b)), as such a point would form a triangle of area less than $\alpha$ with the two fixed points in the corresponding variable segment.
After placing all the  line segments and  fixed points, the area of the smallest possible  triangle will determine the value $\alpha$. 

 

Let $L$ be the set of imprecise points including all the fixed points  and all the segments in the variable gadgets. 
Now an assignment to the variables to satisfy the SAT instance can be
made if and only if a solution for maximizing the area of the minimum possible area triangle of $L$ of area $\alpha$ exists.  
In the following we show how to ensure that all coordinates in $L$ are polynomial, and we prove that the reduction takes polynomial time.
 
\subsection{Correctness of polynomial-time reduction} \label{construction} 
We use an incremental construction method, such that when inserting a new fixed point or the endpoints of a line segment, we do the placement in such a way that the new point is not collinear with two other previously inserted fixed points or the endpoints of two other line segments.
 Also each new  precise  point or the endpoints of a new line segment should not be located in the strips  which are determined by the endpoints of previously inserted line segments. 

Suppose we have $n$ variables and $m$ clauses in the SAT instance.  Since each variable can occur at most one time in each clause,  we need to place  $O(nm)$ fixed points on the lines connecting the common variables in the clause gadgets.  Also we need to place $O(nm)$ other fixed points close to the center of the segments in  the variable gadgets. Therefore we observe that a value for $\alpha$ of $O(\frac 1 {nm^2})$ suffices. 

It is easy to observe that for a given set of $k-1$ points in the plane, we can insert the $k$th point in $O(k^2)$ time, such that  no zero-area triangle exists in the resulting set, and the  new point is not placed in any of the $O(k)$ strips (of the variable gadgets).
We conclude that our reduction can be performed in $O(m^6n^6)$ time in total. 

\begin{theorem}Given a set of $n$ parallel line segments,
the problem of choosing a point on each line segment such that the area of the smallest possible triangle
of the resulting point set is as large as possible is NP-hard.
\end{theorem}


\newpage
\section{\mima problem} \label{sec:minmax}
In this section, we will consider the following problem: given a set $L=\{l_1,...,l_n\}$ of parallel line segments, choose a set $P=\{p_1,...,p_n\}$ of points, where  $p_i \in l_i$,  such that the size of the largest-area triangle with corners at $P$ is as small as possible among all possible choices of $P$ (see Figure~\ref{mnmxmxmn}(d)).

In the following we  first give some definitions, then we state our results.  Let the \textit{top chain}  denote  the lowest convex chain connecting the lower  endpoint of the leftmost line segment to the lower 
endpoint of the rightmost line segment, and which passes over or through the lower  endpoints of the other line segments. 
In other words, it is the top half of the convex hull of the lower endpoints of the segments.
 Similarly, we define the \textit{bottom chain} as the  highest convex chain connecting the upper  endpoint of the leftmost line segment to the upper 
 endpoint of the rightmost line segment, and which passes under or through the upper  endpoints of the other line segments. 
These two convex chains can either be disjoint or intersect and enclose  a  \textit{convex body region} (see Figure~\ref{fig:regionsoftopchain}(a)).
The convex body region can be found  in $O(n \log n)$ time.
 If this region is empty, a single line segment will pass through all the segments, and the smallest largest-area triangle will have  zero area by placing all points collinear.
 This case can be distinguished  in linear time~\cite{edelsb}.
From now on, we assume the convex body region  has non-zero area. Note that one of the top chain or bottom chain can also  be a single  line segment. By \emph{extreme} line segments we mean the leftmost and rightmost line segments.
\begin{figure}
	\includegraphics{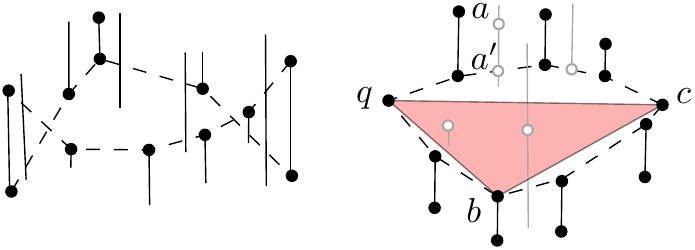}
	\centering
	\caption {(a) Top  chain and bottom chain. (b) For any line segment in $L\setminus L'$ (gray line segment),  any arbitrary point interior to the convex hull is chosen as the candidate point.} 
	\label{fig:regionsoftopchain}
\end{figure}

We first make some observations in Section~\ref{sec:GO}. Then in Section ~\ref{sec:fixed} we solve the special case where the extreme segments are single points. Finally we extend the solution to the general case in Section~\ref{sec:GC}.

\subsection{General Observations} \label{sec:GO}
We start  with some observations. Suppose we have a set of line segments and we know what is the candidate point\footnote{For each $l_i$, the point $p_i \in l_i$ in set $P$ is the candidate point of $l_i$.} on each segment. It is easy to observe that if we add more line segments, the changes on the size of the area of the smallest largest  triangle  is non-decreasing.  

\begin{lemma} \label{lem:intersect}
	Suppose we are given a set $L$ of line segments and we know the set $P$ which consists of the candidate points on $L$ to make the largest-area triangle as small as possible. Then  we insert another line segment $l$, such  that $l$  has a point $x$ which is located on a line segment  passing through  two  points of $P$. Then, $x$ will be chosen as the candidate point of $l$ and this selection will not increase the area of the smallest largest-area triangle.
\end{lemma}

\begin{proof}
Suppose the lemma is false. Then $x$ is increased the size of the smallest largest-area triangle. It follows that $x$ is a vertex  of  the smallest largest-area triangle (Also notice that the largest-area triangle selects all its vertices on the convex hull; no matter we want to minimize its area). Suppose $b$ and $c$ are the other vertices of the smallest largest-area triangle.
Since  $x$ is located on  a line segment  passing through two points of $P$, it is located on a chord on an edge of $CH(P)$. 
 First suppose $x$ is located on a chord of $CH(P)$. We consider a line $\ell$ through $x$ and parallel to $bc$. If we sweep $\ell$ away from $bc$, it will intersect the convex body region until it leaves it at a vertex $x'$. But then $x'$ can be substituted for $x$ to give us a larger smallest largest-area triangle. Contradiction with the previous selection of the smallest largest-area  triangle. 
 
 Now suppose $x$ is  located on the boundary of   $CH(P)$. 
 Again consider the line $\ell$ through $a$ and parallel to $bc$. If we sweep $\ell$ away from $bc$, it will intersect the convex body region until it leaves it at a vertex $x'$. But then again $x'$ can be substituted for $x$ to give us a triangle which its area is either  greater-than or equals to the area of $xbc$. It is again a contradiction with the previous selection of the smallest largest-area triangle. 
  Thus $x$ cannot increase the size of the smallest largest-area triangle.
   
\end{proof}
We will observe later that  the smallest largest-area triangle is always a true triangle.

\begin{corollary} \label{rol}
	Suppose we are given a set $L$ of line segments and we know the set $P$ which consists of the candidate points on $L$ to make the largest-area triangle as small as possible. If a new line segment $l$ intersects the convex hull of  $P$  without sharing a vertex on it, then $l$ does not have a role in the construction of the smallest largest-area triangle.
\end{corollary}

\subsection{Points as the extreme regions} \label{sec:fixed}
In the case where the leftmost and  rightmost line segments are fixed points,  these fixed points  always appear on the convex hull of  $P$. Also, for some line segments that intersect the convex hull of $P$ \footnote{Note that $CH(P)$ consists of interior and boundary.} without sharing a vertex on it, they  share their candidate  points somewhere on or within the convex hull  (see Figure~\ref{fig:regionsoftopchain}(b)).  Let $L' \subseteq   L$  be the set of  line segments from $L$ that share a vertex on the convex body region.  
In the following, we  will prove that the smallest largest-area triangle of $L'$ equals the smallest largest-area triangle of $L$.
 
\begin{obs} \label{obs:intersect}
Any $l_i \in L \setminus  L'$ intersects the convex body region.
\end{obs}

\begin{proof}  Suppose this is not the case. First suppose the convex body region has non-zero area. Then there exists a segment $l_i $  that does not intersect the convex body region. The segment $l_i$  has to  be located completely  above (or below) the convex body region. But then $l_i$ has to share a vertex on the top chain (or the bottom chain), contradiction.

\end{proof} 

%


 As an immediate consequence, we can throw away all the regions in  $L' \subseteq   L$.  In the following, we  show that the largest-area triangle which is inscribed in the convex body region is the smallest largest-area triangle of $L$. In other words, we will prove that in the  case  where the leftmost and rightmost line segments are fixed points, the convex body region is the convex hull of the candidate points.

\begin{lemma} \label{lem:endp}
	Suppose the leftmost and rightmost line segments are  points (degenerate segments), and the convex body region has  non-zero area. Only the upper endpoints on the bottom chain or the lower endpoints  on the top chain  are  candidates for including  the  vertices of the smallest largest-area triangle of $L'$.
\end{lemma}

\begin{proof}
	Suppose the lemma is false. Let $abc$ be the smallest largest-area triangle of $L'$, that selects one of its vertices, $a$,  at a point on one of the line segments of $L'$, so that $a$ is not a lower endpoint on the top chain, or an upper endpoint on the bottom chain, or a fixed point.
	W.l.o.g, assume $l_a$ shares a vertex on top chain. Let $a'$ be the lower endpoint of $l_a$ (see Figure~\ref{fig:regionsoftopchain}(b)). We will prove that we can substitute $a'$ for $a$ without increasing the size of the smallest largest-area triangle.

	We consider a line $\ell$  through $a$ and parallel to $bc$. If we sweep $\ell$  toward $bc$, it will intersect the  top chain. Let $q$ be the first vertex that $\ell$ visits on the top chain.   Suppose  $q \ne a'$, because otherwise we are done.
	Suppose $q \notin \{l_b,l_c\}$. Now we will prove that if we substituted   $a'$ for $a$, the area of  no other  triangles with a vertex at $a$ would be increased.  
	
	If exactly one  of $b$ or $c$ belongs to the bottom chain and  the other one belongs to the top chain, or both  $b$ and $c$ belong  to bottom chain, obviously  moving $a$ to $a'$ will reduce the area of $abc$, and $qbc$ is smaller than $abc$, which contradicts the optimality of $abc$.
	
	Now let  both of $b$ and $c$ belong to the  top chain, then the largest-area triangle  that we want to minimize its area  will never select its  third vertex $a$ on top chain, unless one of  $a$,  $b$ or $c$  is an intersection point of the top chain and bottom chain. 
	Note that point  $a$ cannot  be a  fixed point, because otherwise we have a contradiction. It follows that $a'bc \leq abc$, and also $qbc \leq abc$, which contradicts the optimality of $abc$. Note  that  the smallest largest-area triangle has to be a true triangle (we will discuss it later). 
	
	Now suppose $q \in \{l_b,l_c\}$. Since $a$ is located above the convex body region,  $q$ is the first visited point during  sweeping $\ell$ toward $bc$.  Thus  $l_a$ must intersected an edge of the convex body, contradicting the fact that $l_a$ shares  a vertex on the convex body.
\end{proof}

\begin{corollary} \label{lem:regions}
Suppose the leftmost and rightmost line segments are  fixed points and the convex body region has  non-zero area. Only the fixed points, the lower endpoints on top chain and the upper endpoints on bottom chain  are  candidates  for including the  smallest largest-area triangle's vertices.
\end{corollary}

\begin{lemma} \label{landdlp} 
The smallest largest-area triangle of $L'$ is equal to the smallest largest-area triangle of $L$.
\end{lemma} 

\begin{proof}
From Lemma~\ref{lem:endp} and Observation \ref{lem:regions} we know the smallest largest-area triangle  selects its vertices from the endpoints of  $L'$, and the convex body region is the convex hull of the candidate points on $L$ to minimize the size of the largest-area triangle.  
As we saw in Observation~\ref{obs:intersect},   the  line segments in $ L \setminus  L'$ always intersect the convex body region and they do not have a role in the construction of the smallest largest-area triangle. Thus, from Lemma~\ref{lem:intersect} and Corollary~\ref{rol}  we can throw away the set of line segments in $ L \setminus  L'$.
  If the convex body region has non-zero area, obviously we cannot reduce the size of the largest-area triangle that is inscribed in it. Also from \cite[Theorem 1.1]{48} we know the largest-area  triangle will select its vertices on the convex hull, and it does not the matter that we want to minimize its area. Consequently the smallest largest-area triangle of $L$ equals to the smallest largest-area triangle of $L'$. 

\end{proof}

As an immediate consequence of Lemma~\ref{landdlp}, since all the vertices of $L'$ are selected from distinct segments, the smallest largest-area triangle is always a true triangle. 
 
\subsubsection{Algorithm}
Suppose we are given set $L$ of line segments. After sorting $L$, we can compute the top and bottom chains in linear time. Also, we can compute   $L' \subseteq L$, the convex body region and also its vertices, $P'$, in linear time.   For each  $l_i \in L\setminus L'$, we can choose $p_i \in l_i$ to be an arbitrary point within the convex body region (or on its boundary). Obviously all such points can   be found in linear time. Then we can apply any existing algorithm to compute the largest-area inscribed triangle of $CH(P')$. 
This procedure is outlined in Algorithm~\ref{alg:fixedlr}.

\begin{algorithm}[H]
\caption{\mima  where the extreme line segments are single points}
	\label {alg:fixedlr}
	{\bf Input} {$L=\{l_1,\ldots,l_{n}\}$: a set of vertical segments}\\
	{\bf Output} {$MinMax$:smallest largest-area true triangle}\\
$I$=convex body of $L$ \\
$P'$=vertices of $I$\\
	\Return\textsc {Largest-Triangle($P'$)}\\
	
\end{algorithm}
 
The procedure  \textsc {Largest-Triangle($p'$)} is the  algorithm presented in~\cite{chandran}  which can compute the largest inscribed triangle in linear-time.

\begin{theorem}
Let $L$ be a set of $n$ parallel line segments with two fixed points as the leftmost and rightmost line segments. The solution of \mima problem on $L$ can be found in $O(n \log n)$ time.
\end{theorem}

\subsection {Line segments as the extreme regions} \label{sec:GC}

In this section we study the case where the leftmost and rightmost line segments are not  fixed points. Note that the extreme line segments do not necessarily share a vertex  at the endpoints or even on the optimal solution (see Figure~\ref{notendp}(c,d)). 

Two  polygons that are  constructed by the leftmost and rightmost line segments and  the parts of the top chain and bottom chain which do not contribute to the boundary of the  convex body region  (hatched regions in Figure~\ref{topbot}(a)) 
are called the  \textit{tail  regions}. In this case, $L'$ also includes the line segments which share vertices on the tail regions.  In a similar fashion,   
$L \setminus L'$  consists of the line segments that have two intersection points with the tail regions without sharing a vertex on it, and also the line segments that are intersected by the convex body region without sharing a vertex on it.
\begin{lemma}
Let $L$ be a given set of imprecise points. Let $P$ be the  set of points that minimizes the size of the largest-area triangle.  Then $CH(P)$ always intersects  the line segments in $L \setminus L'$. 
\end{lemma}
\begin{proof}
Suppose the lemma is false. Then there exists a line segment $l_p \in L \setminus L'$, such that $l_p$ is not intersected by $CH(P)$. 
First notice that $CH(P)$ includes the convex body, and consequently it is intersected by all the line segments which are intersecting by the convex body without sharing a vertex on it. 
Thus $l_p$ must be  a line segment which is intersecting a tail region, say the right tail,  while it does not share any vertex on it.  
 Notice that $CH(P)$ selects also some vertices on the leftmost and rightmost  line segments. Two cases can happen. First, suppose $CH(P)$  selects an endpoint of $l_r$.
 W.l.o.g., suppose it is selected the upper endpoint of $l_r$, but then  it must also includes all the vertices of the upper concave chain of the right tail (we discuss it in next lemma), 
  and in which case it must intersect $l_p$ which gives a  contradiction.
 
 Now consider the case where $CH(P)$ selects a point somewhere on the middle of $l_r$. In this case again it has to be intersected by $l_p$. Contradiction. 
\end{proof}
Consequently, with the same argument we had  in Section~\ref{sec:fixed}    the candidate point of  any line segment with two intersection points with a tail region (without sharing any vertex on it) can be chosen to be any point within or on the boundary of the convex hull of candidate points.
For this, it suffices to find the candidate points of the leftmost and rightmost segments. 
Notice that with the same argument that  we have proven  in Lemma~\ref{lem:endp}, all the vertices of the convex body and tail regions should be involved in the convex hull of the candidate points.  
In Lemma~\ref{minma} we  will prove that for the line segments in $L'\setminus \{l_l,l_r\}$ which share some vertex on the tail regions, there is no other interesting point to be  candidate for a vertex of the smallest largest-area triangle.

\begin{figure}
	\includegraphics{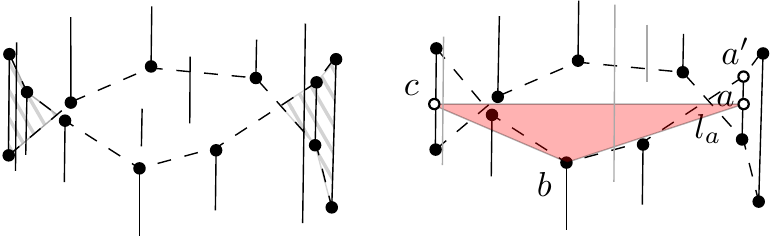}
	\centering
	\caption {(a) Hatched regions illustrate the tail regions. (b) For any segment in $L'$, only the shared vertices are the candidates of including the vertices of the smallest largest-area triangle.
		 }
	\label{topbot}
\end{figure} 

\begin{lemma}
\label{minma}
Suppose the convex body region  has non-zero area. For any  line segment in $L'\setminus \{l_l,l_r\}$ only the shared vertices  are  candidates for the vertices of the smallest largest-area triangle. 
\end{lemma}

\begin{proof}
Suppose the lemma is false, then the smallest largest-area triangle $abc$  selects a vertex $a$ on $l_a$, such  that $l_a$ shares some vertex on e.g., the right tail region. 
First suppose  $a$  is located interior to the right tail region. W.l.o.g, suppose $l_a$ shares  some vertex $l_a^-$ on the bottom chain. Let $l_a^+$ be the upper endpoint of $l_a$ and let $a$ be any point of $l_a$  with  a lower $y$-coordinate than $l_a^+$ and interior to the right tail region (as illustrated in Figure~\ref{topbot}(b)).
We know that the convex hull of the candidate points has a vertex $r$ on $l_r$, and  $abc$ is the largest-area  true triangle among other possible triangles. Consider a line $\ell$ through $a$ and parallel to $bc$.  If we sweep $\ell$ toward $bc$ we will visit $r$. It is easy to observe that removing vertices from the convex hull of a set of points may only decrease the area  of the largest-area triangle that is inscribed in the convex hull (notice  that the smallest largest-area triangle is the largest triangle which is inscribed in the convex hull of the candidate points). By the optimality of the area of $abc$ we know  that all the other triangles which are rooted at $a$ are not larger than $abc$. Thus we just omit the vertex $a$ of the convex hull,  while the convex hull  is still  intersected by $l_a$ and by all the other segments which were intersected by the removed  ear of the hull, since the removed ear is completely located within the right tail region. Thus $abc$ could not be the smallest largest-area triangle. Contradiction.



Now suppose $a$ is located  outside the  tail region and w.l.o.g above it. Suppose $l_a$ shares a vertex $l_a^-$ on the bottom chain (obviously if it also shares $l_a^+$, $l_a^+$ is a better choice with respect to $a$ and we are done). In which case again $l_a$ is intersected by the convex hull of the candidate points. Once again by removing the vertex $a$ of the convex hull,   the convex hull  is still  intersected by $l_a$ and by all the other segments which were intersected by the removed  ear of the hull.  Thus again $abc$ could not be the smallest largest-area triangle. Contradiction.
\end{proof}

\begin{figure}
	\includegraphics{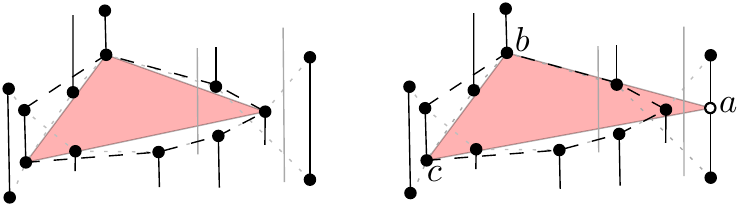}
	\centering
	\caption {(a) The smallest largest-area triangle which is inscribed in the convex hull (dashed polygon) of the top and bottom chain excluding ${l_l,l_r}$.
		(b) Triangle $abc$ is the solution of \mima  which selects one vertex on the rightmost  segment. }
	\label{lefrig}
\end{figure}

\subsubsection{Algorithm}
Suppose we are given a set $L$ of line segments. After sorting $L$ we can compute the top and bottom chain in linear time. Also, we can compute the set $L'$, the convex body region, the tail regions and also their vertices $P'$  in linear time. 
 In the following, we consider the procedure of finding the solution to the \mima problem of $L'$ in several possible configurations.  

We proceed by parameterizing the problem in the location of the points on the leftmost and rightmost line segment. The area of each potential triangle is linear in these parameters, yielding a set of planes in a $3$-dimensional space. We are looking for the lowest point on the upper envelope of these planes; see Figure~\ref {arran}(b).
  In the case where only one of the leftmost or rightmost line segments share a vertex on the optimal solution, $O(n^2)$ triangles  may have a common third vertex on the leftmost or rightmost line segments; see Figure~\ref {arran}(a).
  Since  the lowest vertex of the upper envelope can be found in  $O(n^2 \log n)$ time, the solution of the \mima problem in this case can be computed in $O(n^2 \log n)$ time. 

From Lemma~\ref{minma} we know that for each line segment $l_i$ $\in$  $L\setminus \{L' \}$, we can choose an arbitrary point  $p_i \in l_i$ within the convex hull of the convex body and the determined candidate points of the extreme segments.

If the smallest largest-area triangle does not use the leftmost and rightmost line segments, by the same argument as in Lemma~\ref{landdlp}, the smallest largest-area triangle is inscribed in the convex hull of the top and bottom chain excluding the extreme segments (as illustrated in Figure~\ref{lefrig}). But since  here the convex hull of the candidate points is not necessarily a true convex hull (see Figure~\ref{lefrig}(a)),  we apply Algorithm~\ref{alg:impquad}. This procedure is outlined in Algorithm~\ref{alg:fixedlr}.

\begin{algorithm}[H]
	\caption{\mima with one vertex at extreme segments}
	{\bf Procedure} {\sc One-ExtrSeg($l_e$)} \\
	{\bf Input} {$L=\{l_1,\ldots,l_{n}\}$: A set of vertical segments}\\
	{\bf Output} {: smallest largest-area true triangle with one vertex at $l_e$}\\
	$I$=convex body of $L$\\
	$T_l,T_r$=tails of $L$\\
	$P'$=vertices of  $CH(I \cup T_l \cup T_r \setminus \{ \{l_l,l_r\}\setminus l_e\})$\\
	$\textbf{for}$ $\textbf{each}$ pair $x,y \in P'$\\
	\hspace{0.5cm} compute the area function $f(x,y)$ with third vertex at $l_e$\\
	\hspace{0.5cm} add $f(x,y)$ to an arrangement $A$\\
	$\textbf{end}$\\
	$v$= the lowest vertex of the lower envelop of $A$\\
	$Area(v)$=the determined area on the arrangement $A$ by $v$\\
	\Return  $Area(v)$ \\
\end{algorithm}

\begin{algorithm}[H]
	\caption{\mima with two vertices at extreme segments}
	{\bf Procedure} {\sc Two-ExtrSeg($l_l,l_r$)} \\
	{\bf Input} {$L=\{l_1,\ldots,l_{n}\}$: a set of vertical segments}\\
	{\bf Output} {: smallest largest-area true triangle with two vertices at $l_l$ and $l_r$}\\
	$I$=convex body region of $L$\\
	$P'$=vertices of  $CH(I)$\\
	$\textbf{for}$ $\textbf{each}$ $x \in P'$\\
	\hspace{0.5cm} compute the area function $f(x)$ with two vertices at  $l_l,l_r$\\
	\hspace{0.5cm} add $f(x)$ to an arrangement $A$\\
	$\textbf{end}$\\
	$v$=lowest vertex of the lower envelop of $A$\\
	$Area(v)$=the determined area on the arrangement $A$ by $v$\\
	\Return  $Area(v)$ \\
\end{algorithm}

\begin{algorithm}[H]
	\caption{\mima  where the extreme line segments are not single points}
	\label {alg:fixedlr}
	{\bf Input} {$L=\{l_1,\ldots,l_{n}\}$: a set of vertical segments}\\
	{\bf Output} {: smallest largest-area true triangle}\\
	$U$=top chain of $L$\\
	$L$=bottom chain of $L$\\
	$I$=convex body of $U$ and $L$\\
	$T_l,T_r$=tails  of $L$\\
	$P'$=vertices of  $CH(I \cup T_l \cup T_r \setminus \{l_l,l_r\}$)\\
	\Return\textsc Max({\sc Largest-Inscribed-Triangle($P'$)},	{\sc One-ExtrSeg($l_l$)},	{\sc One-ExtrSeg($l_r$)},	{\sc Two-ExtrSeg($l_l,l_r$)})\\
	
\end{algorithm}

\subsection{Correctness Proof}
 The correctness of the algorithm comes from the correctness  of Lemma~\ref{minma},  which states that we can find the optimal position of the candidate points on the line segments which are located between the leftmost and  rightmost line segments, and  on the leftmost and the rightmost line segments as well.
 
 If the smallest largest-area triangle does not use  the leftmost and  rightmost line segments, changing the position of  its vertices on their segments may only increase its area. 
If   the smallest largest-area triangle does use the leftmost and/or rightmost line segments, as it is the largest possible triangle among all the other possible triangles, its area determines the optimal position of the vertices on the leftmost and/or rightmost line segments.  Thus, changing the position of its vertices on the leftmost and/or rightmost segments can just increase its area. Such points are in balance between the area of at least two triangles, they could move in one direction to decrease the area of one triangle but only by increasing the area of another triangle. Thus, the algorithm finds the optimal solution.

\begin{figure}
 	\centering
\includegraphics{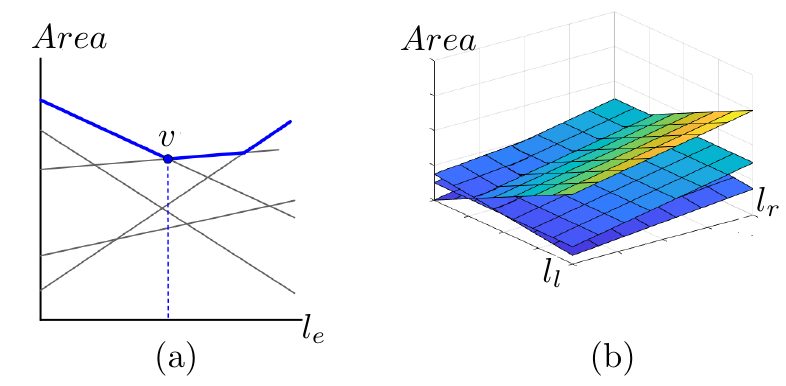}
\caption{ (a) Image of  $v$ on $l_e$ determines the point that minimizes the area of all possible triangles with one vertex on $l_e$, where $l_e \in \{l_l,l_r\}$. (b) A vertex  on the upper envelope of  the  half planes determines the point that minimizes the area of all possible triangles with two vertices on $l_l$ and $l_r$.}
\label{arran}
\end{figure}

\begin{theorem} Let  $L$  be a set of $n$ imprecise points modeled as parallel line segments. The solution of the problem \mimi can 	be found in $O(n^2 \log n)$ time.
\end{theorem}

\newpage
\section{ Extension to  $k>3$}  \label{sec:generalk}
 It is natural to ask the same questions posed in Section~\ref{sec:problemdef} for polygons with $k>3$ sides, in which case we look for tight lower bounds and upper bounds on the area of the smallest or largest convex $k$-gon  with corners at distinct line segments. 
\begin{figure}
	\includegraphics{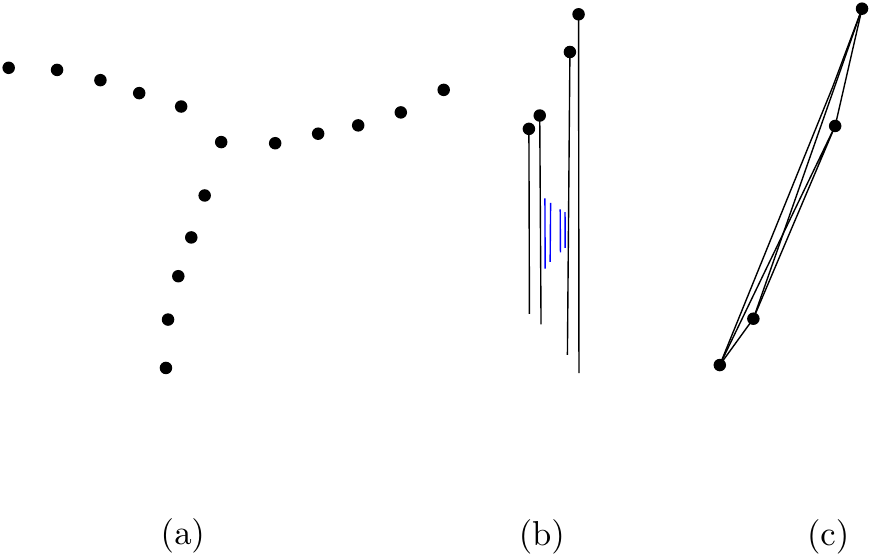}
	\centering
	\caption {(a)  A set $P$ of $n=3k$  points, in  which there is no convex  $k$-gon determined by a subset  of $P$ with any value of $k \ge 5$. (b) An example of a  set  of parallel line segments, in which at most $k=4$ points  can be selected at the endpoints to make a convex $k$-gon. Note that the blue short line segments at the middle of the long line  segments are rescaled copies of the long  segments. This configuration can be repeated arbitrarily. 
		(c) The close-up of   upper endpoints of  long line segments.  In any configuration of the same color four line segments, at most two endpoints  at the lower or upper endpoints can contribute to the boundary of a  convex $k$-gon.   }
	\label{lesskgon}
\end{figure}
Note that even in the precise version of finding an optimal area $k$-gon, not every point set will contain a convex $k$-gon for all values of $k$, as illustrated in Figure~\ref{lesskgon}(a).
Clearly, the same is true for the imprecise version, since the points in  Figure~\ref{lesskgon}(a) can be a set of very short line segments. 



In addition, note that by insisting on using exactly $k$ vertices in our convex polygon, we may introduce unexpected behavior, forcing some vertices to lie in the interiors of their segments just to ensure there are sufficiently many vertices in convex position,  as illustrated in   Figure~\ref{lesskgon}(b). 


The question of whether  are there $k$ points in convex position or not is studied in~\cite{kgonsolved}, where the question was ``given a set $L=\{l_1,...,l_n\}$ of parallel line segments,
 decide whether or
not there exist $k$ points on distinct segments, such that all the points are in convex position''. Similarly, this problem is also studied for many other geometric planar objects in the same paper. 
The authors showed that for a set of parallel line segments, the above decision question can be answered in polynomial time, but for general regions, the problem becomes NP-hard. 

As an immediate consequence of our present results, e.g.,   Section~\ref{sec:alen},  the above decision question for some small values of $k$, (for example $k=3$) can be answered in polynomial time;
however, in general, the problem appears to quickly become difficult.

Nonetheless, we take the apparent difficulty of the problem and its unexpected results as motivation to relax the problem definition; in Sections~\ref {sec:kmama} and~\ref {sec:kmima} we study the largest convex polygon with {\em at most} $k$ vertices. Note that asking for the {\em smallest} convex polygon on at most $k$ vertices is not a sensible question; for any point set, the smallest $m$-gon   ($m<k$) has area $0$, since we chose $m=2$. So the maximum over all possible sets is still $0$.

 \subsection{\mama problem for polygons with $k>3$ sides} \label {sec:kmama}
 In this section we  study the following problem.

\emph{\mama}:  given a set $L=\{l_1,...,l_n\}$ of parallel line segments, and an integer $k$ ($k\le n$), choose a set $P=\{p_1,...,p_n\}$ of points, where  $p_i \in l_i$,  such that the size of the largest-area polygon $Q$ with corners at $P$ is as large as possible among all  choices of $P$.
See Figure~\ref{instancemaxmax} for an example. 

\begin{lemma} 
There is an optimal solution to  \mama for a convex polygon with at most $k$ vertices, such that all vertices are chosen at the endpoints of  line segments. 
\end{lemma}

\begin{proof}
Suppose the lemma is false. Then there exists an    $m$-gon ($m \leq k$) $Q$  with maximum possible area, and a minimal number of vertices that are not at the endpoint of their line segments.    Any vertex $p$ of $Q$  which is not located at an endpoint of its segment can be moved to one of its endpoints and increase the area of $Q$. It is possible that   $Q$ is no longer convex or some vertices of $Q$ lie  no longer  on the boundary of $Q$. But correcting these by removing these points from $Q$ can only increase the area of $Q$, contradicting the  choice of $Q$.
\end{proof}


\begin{figure}
	\centering
	\includegraphics{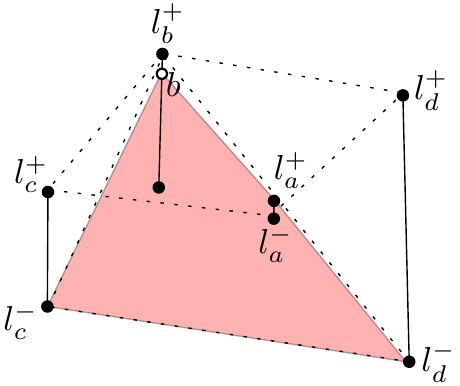}
	\caption { For $k=4$, $l_a^+bl_c^-l_d^-$ is the largest-area $k$-gon, where the  inner angle at $l_a^+$ approaches $\pi$, and $Q=l_b^+l_c^-l_d^-$ is the largest-area convex polygon with at most $k$ vertices. Note that there are some  $k$-gons which can be constructed at the endpoints (e.g., $l_b^+l_c^+l_a^-l_d^+$), but they were not reported since they are smaller than $Q$. }
	\label{instancemaxmax}
\end{figure}

\subsubsection{Algorithm} \label{mamakgon}
In the following, we provide a  dynamic programming algorithm for \mama. 
For $p \neq q$,  we define an array $A[l_p^-,l_q^{+},m]$ which denotes the area of the largest  true $m$-gon   with the bottommost vertex $l_p^-$, such that $l_q^{+}$ is the next vertex of $l_p^-$  on the counterclockwise ordering of the boundary of $m$-gon (see Figure~\ref{dp}).  This $m$-gon is constructed on the  largest-area $m-1$-gon $A[l_p^-,l_r^{+},m-1]$  with the bottommost vertex $l_p^-$,  where the the  triangle $l_p^-l_q^{+}l_r^+$ is the largest-area triangle among all possible choices for $l_q$.


 For finding the optimal solution, we should compute all such $m$-gons, where $m \leq k$, and the largest-area polygon for some value of $m$ will determine the largest-area convex polygon with at most $k$ vertices.   

  Note that  our dynamic programming  can also report whether there exists a convex polygon with exactly $k$ vertices to be constructed at the endpoints of the line segments~\footnote{If there is no  candidate for $l_q$ to increase the area of the current $m$-gons for some values of $m<k$, then there is no larger $m+1$-gon to be constructed at the endpoints.
  	}. As said above,  in the original setting of our dynamic programming, we will report the largest convex polygon with at most $k$ vertices.

Let  $H_{l_p^-,l_q^+}$ be the half plane to the left of supporting line of $\overrightarrow{l_p^-l_q^+}$. We sort for each point $l_r^+$, all the other endpoints in clockwise order around $l_r^+$ and store these orderings. Computing the clockwise order of all points around any point in a set  can be done in quadratic time totally~\cite{Edel}.
For any $l_p^-$ and $l_r^+$ we  consider all the line segments $l_q$  which have an endpoint that is sorted clockwise around the point $l_r^+$. If $l_r^+$ is located to the right of the supporting line of  $\overrightarrow{l_p^-l_q^+}$, then no new point can be used.  
The decision about the candidate endpoint of $l_q$ will be made according to the  direction of the edge $l_p^-l_r^+$  (see Figure~\ref{dp}). 

Note that $A[l_p^-,l_r^-,m]$  is  also  considered. We suppose (w.l.o.g) the optimal $m$-gon selects the endpoint  $l_r^+$.  
 We will start by initializing all the values of the array by zero, as the area for  any $m=2$ is zero. Then we have


$$A[l_p^-,l_q^+,m]=max_{l_r^{+} \in H_{l_p^-,l_q^+}}(A[l_p^-,l_r^{+},m-1]+area(l_p^-l_q^+l_r^+))$$

Moreover, for any computation $A[l_p^-,l_q^+,m]$ for $m>3$, we check  for convexity, that may reduce the value of $m$. 
Note that  if there exists a non-convex at most $m$-gon with area $A$, then there also exists a convex polygon with area more than $A$ and still at most $m$ vertices. 

 It is obvious that the above recursion can be evaluated in linear time for each value of $l_p^-$, $l_q^+$ and $l_r^+$. Thus it will cost $O(kn^3)$ time and $O(n^2)$ space totally. 

\begin{theorem} Let  $L$  be a set of $n$ imprecise points modeled as parallel line segments. The solution to the   \mama problem for a convex polygon with at most $k$ vertices can be found in $O(kn^3)$ time.
\end{theorem}

 \begin{figure}
	
	\includegraphics{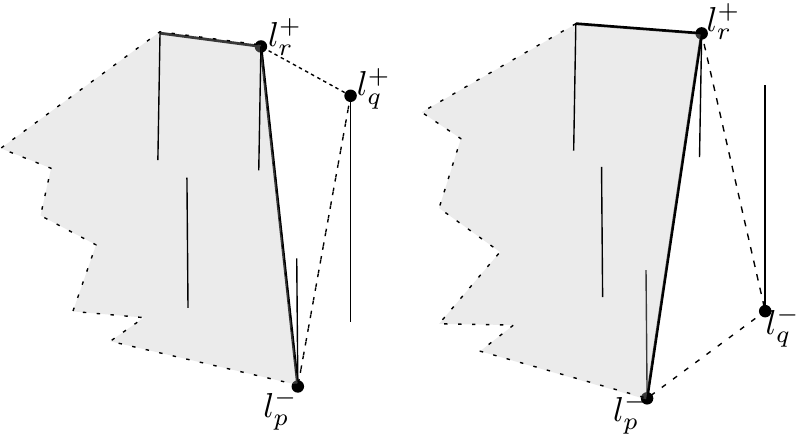}
 	\centering
 	\caption {The directed segment  ${l_p^-l_r^+}$  indicates  which of the upper or  lower endpoints of $l_q$ should be selected.}
 	\label{dp}
 \end{figure}

\subsection{\mima problem for polygons with $k>3$ sides} \label {sec:kmima}
In this section we study the following problem:

	\emph{\mima}:  given a set $L=\{l_1,...,l_n\}$ of parallel line segments, and an integer $k$ ($k \le n$), choose a set $P=\{p_1,...,p_n\}$ of points, where  $p_i \in l_i$,  such that the size of the largest-area polygon $Q$ with corners at $P$ is as small as possible among all  choices of $P$, and we extend the approach from Section~\ref{sec:minmax}. Again, we first solve the problem for the special case where $l_l$ and $l_r$ are single points. Then in Section~\ref{sec:GCk-gon} we extend the solution to the general case.

\subsubsection{Single points as the extreme regions}

For the \mima problem,  when the leftmost and rightmost line segments are fixed points, the solution of \mima with at most $k$ vertices cannot  have an area smaller than the area of the largest  at most $k$-gon which is inscribed in the convex body region, since moving the vertices of the solution among  their line segments can only increase the area.

 Similar to the case $k=3$, in this case also the solution is always a true convex polygon.\footnote{Note that  a zero area solution can only occur when the convex body region has zero area. As said before, this case can be distinguished in linear time~\cite{edelsb}.}.

 Thus,   the  algorithm presented by Boyce \etal~\cite{48} can be applied to find the largest inscribed $k$-gon in the convex body region. This algorithm runs in $O(kn+n \log n)$ time. 
If the convex body region has  $m < k$ vertices, clearly  it is not possible to minimize the size of this $m$-gon any more, and thus the  convex body region is the smallest largest-area polygon with at most $k$ vertices.

\begin{theorem} Let  $L$  be a set of $n$ imprecise points modeled as parallel line segments, where the leftmost and  rightmost line segments are  points. The solution of the problem \mima for a convex polygon with at most $k$ vertices can 	be found in $O(kn+n \log n)$ time.
\end{theorem}

 \newpage
\subsubsection{Line segments as the extreme regions} \label{sec:GCk-gon}
As we saw in Figure~\ref{lefrig}(a), if the leftmost and rightmost line segments are not single points, then the  convex hull of the top chain and bottom chain (without using the endpoints of the extreme line segments) is not necessarily a true convex hull. We called this convex polygon $\delta$. 
  Assume the convex body region has non-zero area.
First suppose  the optimal solution  is inscribed in  $\delta$. If there are no $k$ different   vertices on $\delta$, then the largest-area true convex hull which is constructed on all the  vertices is the smallest largest convex polygon with at most $k$ vertices, as we cannot shrink its area anymore. In this case, the  dynamic program algorithm which is presented in Section~\ref{mamakgon}  will find the optimal solution in $O(kn^3)$ time.

Now we consider the case where the extreme line segments can share a vertex on the optimal solution. It is easy to observe that the method presented for the case  $k=3$   costs $O(n^{k-1} \log n)$ time for general values of $k$. We design a dynamic programming algorithm that can solve the problem in $O(n^8 \log n)$ time. 

 We start the description of the algorithm by a definition. We say a convex chain $C$  is a {\em supporting chain} of a convex   polygon $P$, where $C$ is an ordered subset of the vertices of $P$ in counterclockwise direction, and the   line passes through the endpoints of $C$  supporting the polygon $P \setminus  C$ on  the convex side  of $C$ (see Figure~\ref{sconvexchain}).  
\begin{figure}
	\centering
	\includegraphics{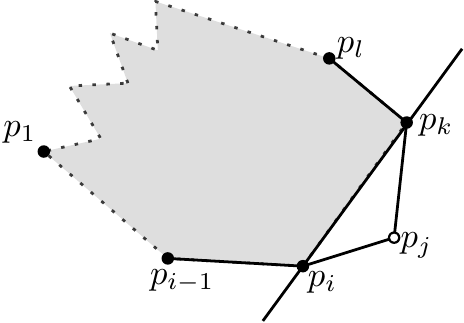}
	\caption {Chain $C=p_ip_jp_k$ supports polygon $P=p_1,p_2,...,p_n$, where the line $p_ip_k$ supports the convex polygon $p_1,...,p_{i-1}p_{i}p_kp_l,...,p_n$ on the convex side of $C$.}
	\label{sconvexchain}
\end{figure}

First suppose the optimal solution selects only one vertex at an extreme line segment, e.g., the  rightmost line segment; the other case is similar.
   Then it suffices to find the largest triangle to the right  of a clockwise convex chain $C$  consisting of three line segments, such that $C$ supports the largest-area convex polygon with at most $k-1$ vertices at the convex side of $C$ (see Figure~\ref{sconvexchain}).  

Indeed for any $k$ we look for a counterclockwise chain  $C$ with  three line segments  \footnote{In our dynamic program, we  also consider the case where only one or two line segments exist on the convex chain, since the solution might includes only one triangle.} which supports the largest-area polygon with at most $k-1$ vertices to the left of $C$. 
For each chain $C$,  the optimal position of the at most  $k$-th vertex of the optimal convex polygon which is supported by $C$ should be computed. Thus we  write our dynamic programming according to this observation. Let ${p_j}$ be the optimal position of the candidate point at the rightmost line segment. For a counterclockwise ordered set of vertices 
$p_{i-1}, p_i, p_k$  and $p_l$, where $p_{i-1} \neq p_l$  we define the array $A[p_{i-1},p_i,p_k,p_l,m-1]$ which denotes the largest-area convex polygon with at most $m-1$ vertices, which is supported by the convex chain $p_{i-1} p_i p_k p_l$, where $p_k$ is the rightmost vertex of the chain. 

$$A[p_{i-1},p_i,p_j,p_k,m']=max(A[p_{i-1},p_i,p_k,p_l,m-1]+area (p_ip_jp_k), area(p_{i-1}p_jp_l))$$
Note that for any $p_{i-1}, p_i, p_k$ and $ p_l$ we know that they are at the lower endpoints or upper endpoint of their corresponding line segments, since they are the vertices of the top and bottom chain. 
When the convex body region has non-zero area,  the above equation  considers the correct solution  which consists of at least a triangle.
 
Obviously the above equation can be evaluated in $O(mn^4)$ time. The vertex $p_j$ can   be computed in constant time for  fixed $p_i$ and $p_k$.
The triangle $p_{i-1}p_jp_l$ should also be evaluated separately   while  the triangle $p_{i-1}p_jp_l$ may cover the convex chain. In this case $m'=3$, otherwise $m'=m$.

 \begin{figure}
\centering
 	\includegraphics{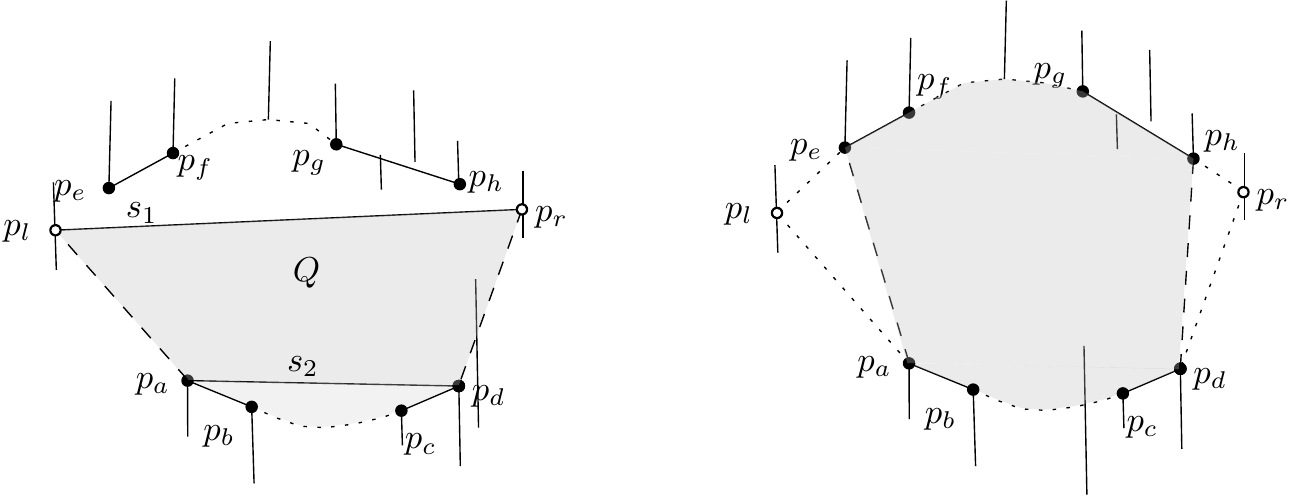}
 	\caption {(a) The smallest largest-area at most $k$-gon  is located below the line segment connecting the extreme  segments. The quadrilateral $p_ap_dp_rp_l$ has a common edge $p_ap_d$ with convex chain $p_cp_dp_ap_b$. (b) Two convex chains $p_cp_dp_hp_g$ and $p_fp_ep_ap_b$ are supporting the smallest largest-area convex polygon with at most $k-2$ vertices.  }
 	\label{dplr}
 \end{figure}

In the case where the optimal solution  can select its vertices at both of the leftmost and the rightmost line segments,  two different cases should be considered:
\begin{itemize}
\item The case where the candidate points of the leftmost and the rightmost line segments are directly connected together via a line segment $s_1$, in which case the optimal solution is located above or below $s_1$. An example is illustrated in Figure~\ref{dplr}(a). Consider a quadrilateral $Q$ with edge $s_1$. Let $s_2$ be the opposite side of $s_1$.
In this case, we should consider all the convex chains  which support a largest-area convex polygon with at most $k-2$ vertices, and have an edge coinciding with $s_2$.
We selected two vertices of $Q$ on the leftmost and rightmost line segments so that we minimize the  area of all possible quadrilaterals  $Q$  on different convex chains.

\item The case where the candidate points at the leftmost and  rightmost line segment are not directly connected together. Thus we should consider  two convex chains which are supporting a largest-area convex polygon with at most $k-2$ vertices from two sides top and bottom, and find the two other vertices at the leftmost and  rightmost line segments so that this selection minimizes the area of the convex polygon with at most $k$ vertices.  An example is illustrated in Figure~\ref{dplr}(b).

\end{itemize}
We  handle these  cases separately. First suppose  the optimal solution is located below the line segment $p_lp_r$ which is connecting the leftmost and the rightmost line segment, as illustrated in Figure~\ref{dplr}(a). Note that the optimal position of the vertices at the extreme line segments should be computed by considering all  the largest-area at most $k-2$-gons that might  be constructed on both sides of $p_lp_r$. 

We start by computing a set of surfaces, such that each surface is constructed on a line segment $p_dp_a$, which determines by connecting the endpoints of  the  convex chain $p_ap_b...p_cp_d$, which supports the  smallest largest-area polygon with at most $k-2$ vertices, as illustrated in  Figure~\ref{dplr}(a). 
This convex chain starts at  the line segments $p_ap_b$ and ends at another line segment $p_cp_d$ in counterclockwise direction.  After computing all such surfaces, we draw all of them as an area function of the positions of the vertices of $p_l$ and $p_r$ (leftmost and rightmost vertices, respectively), such that the lowest vertex in the lower envelope of the surfaces will determine the optimal position of the vertices on the leftmost and the rightmost line segments.

As said above, we also  consider all the possible surfaces which can be constructed on some convex chains above $p_lp_r$, e.g., $p_hp_g...p_fp_e$. In this case, the lowest vertex of the upper envelope can be computed in $O(n^4\log n)$ time.   
It is easy to observe that the computed solution is always true, since all the considered vertices below the line segment $p_lp_r$ are located at distinct line segments.

In the case where the extreme line segments are not connected directly, the optimal position of the vertices at the  extreme line segments can be determined by computing all  surfaces corresponding to convex polygons with at most $k-2$ vertices, which are supported  by two convex chains from left and right,  as illustrated in Figure~\ref{dplr}(b). 
   It is easy to observe that in this case the lowest vertex of the upper envelope of the surfaces can be computed in $O(kn^8\log n)$ time. Note that in the extra $O(k)$ time cost we  check whether the vertices of any potential solution are located at different line segments, and thus the computed solution is always a true convex polygon. 

\begin{theorem} Let  $L$  be a set of $n$ imprecise points modeled as parallel line segments. The solution to the problem \mima for a convex polygon with at most $k$ vertices can 	be found in $O(kn^8 \log n)$ time.
\end{theorem}


\section{Conclusions and open problems}
In this paper we studied smallest and largest triangles on a set of points, a classical problem in computational geometry, in the presence of uncertainty.  
Many open problems  still remain, even aside from the obvious question of whether our algorithms have optimal running times.
The choice of parallel line segments as the imprecision regions already leads to a rich theory which will serve as a first step towards solving the problem for more general models of uncertainty. It is unclear to what extent our results generalize; It is conceivable that some technique (e.g. based on convex hulls) do while others (e.g. based on dynamic programming) do not.
One intriguing open question concerns the \mami problem,  where the question is to find $k$ points on distinct regions, such that the selected vertices maximizes the size of the area of the $k$-gon. We conjecture that this problem is NP-hard even in a discrete model of imprecision (in which each region is a finite set of points). 

\subparagraph*{Acknowledgments}
This work was partially supported by the Netherlands Organization for Scientific Research (NWO) under project no. 614.001.504.

\bibliographystyle{abbrv}


\bibliography{sample}

\begin{thebibliography}{10}

\bibitem{agarwalconvex}
P.~K. Agarwal, S.~Har-Peled, S.~Suri, H.~Y{\i}ld{\i}z, and W.~Zhang.
\newblock Convex hulls under uncertainty.
\newblock In {\em European Symposium on Algorithms}, pages 37--48. Springer,
  2014.

\bibitem{49}
A.~Aggarwal, H.~Imai, N.~Katoh, and S.~Suri.
\newblock Finding $k$ points with minimum diameter and related problems.
\newblock {\em J. of algorithms}, 12(1):38--56, 1991.

\bibitem{msearch}
A.~Aggarwal, M.~Klawe, S.~Moran, P.~Shor, and R.~Wilber.
\newblock Geometric applications of a matrix searching algorithm.
\newblock In {\em 2th Annual Symposium on Computational Geometry}, pages
  285--292. ACM, 1986.

\bibitem{AHN2013253}
H.-K. Ahn, S.-S. Kim, C.~Knauer, L.~Schlipf, C.-S. Shin, and A.~Vigneron.
\newblock Covering and piercing disks with two centers.
\newblock {\em Computational Geometry}, 46(3):253 -- 262, 2013.

\bibitem{56}
E.~M. Arkin, C.~Dieckmann, C.~Knauer, J.~S. Mitchell, V.~Polishchuk,
  L.~Schlipf, and S.~Yang.
\newblock Convex transversals.
\newblock In {\em Workshop on Algorithms and Data Structures}, pages 49--60.
  Springer, 2011.

\bibitem{54}
M.~Atallah and C.~Bajaj.
\newblock Efficient algorithms for common transversals.
\newblock {\em Information Processing Letters}, 25(2):87--91, 1987.

\bibitem{41}
D.~Avis and D.~Rappaport.
\newblock Computing the largest empty convex subset of a set of points.
\newblock In {\em 1th annual symposium on Computational geometry}, pages
  161--167. ACM, 1985.

\bibitem{48}
J.~E. Boyce, D.~P. Dobkin, R.~L.~S. Drysdale, and L.~J. Guibas.
\newblock Finding extremal polygons.
\newblock {\em SIAM J. on Computing}, 14(1):134--147, 1985.

\bibitem{chandran}
S.~Chandran and D.~M. Mount.
\newblock A parallel algorithm for enclosed and enclosing triangles.
\newblock {\em International J. of Computational Geometry \& Applications},
  02(02):191--214, 1992.

\bibitem{72}
O.~Daescu, W.~Ju, and J.~Luo.
\newblock Np-completeness of spreading colored points.
\newblock {\em Combinatorial Optimization and Applications}, pages 41--50,
  2010.

\bibitem{mountproximity}
O.~Daescu, J.~Luo, and D.~M. Mount.
\newblock Proximity problems on line segments spanned by points.
\newblock {\em Computational Geometry}, 33(3):115--129, 2006.

\bibitem{58}
J.~M. D{\'\i}az-B{\'a}{\~n}ez, M.~Korman, P.~P{\'e}rez-Lantero, A.~Pilz,
  C.~Seara, and R.~I. Silveira.
\newblock New results on stabbing segments with a polygon.
\newblock {\em Computational Geometry}, 48(1):14--29, 2015.

\bibitem{newdob}
D.~P. Dobkin, R.~Drysdale, and L.~J. Guibas.
\newblock Finding smallest polygons.
\newblock {\em Computational Geometry}, 1:181--214, 1983.

\bibitem{42}
D.~P. Dobkin, H.~Edelsbrunner, and M.~H. Overmars.
\newblock Searching for empty convex polygons.
\newblock {\em Algorithmica}, 5(1-4):561--571, 1990.

\bibitem{45}
D.~P. Dobkin and L.~Snyder.
\newblock On a general method for maximizing and minimizing among certain
  geometric problems.
\newblock In {\em 20th Annual Symposium on Foundations of Computer Science},
  1979.

\bibitem{52}
R.~L.~S. Drysdale and J.~W. Jaromczyk.
\newblock A note on lower bounds for the maximum area and maximum perimeter
  $k$-gon problems.
\newblock {\em Information Processing Letters}, 32(6):301--303, 1989.

\bibitem{drysdale2008nlogn}
R.~S. Drysdale and A.~Mukhopadhyay.
\newblock An $o (n\log n)$ algorithm for the all-farthest-segments problem for
  a planar set of points.
\newblock {\em Information Processing Letters}, 105(2):47--51, 2008.

\bibitem{edelsb}
H.~Edelsbrunner.
\newblock Finding transversals for sets of simple geometric figures.
\newblock {\em Theoretical Computer Science}, 35:55--69, 1985.

\bibitem{EDELSBRUNNER1989165}
H.~Edelsbrunner and L.~J. Guibas.
\newblock Topologically sweeping an arrangement.
\newblock {\em Journal of Computer and System Sciences}, 38(1):165 -- 194,
  1989.

\bibitem{Edel}
H.~Edelsbrunner, J.~O’Rourke, and R.~Seidel.
\newblock Constructing arrangements of lines and hyperplanes with applications.
\newblock {\em SIAM J. on Computing}, 15(2):341--363, 1986.

\bibitem{51}
D.~Eppstein.
\newblock New algorithms for minimum area k-gons.
\newblock In {\em Proceedings of the third annual ACM-SIAM symposium on
  Discrete algorithms}, pages 83--88. Society for Industrial and Applied
  Mathematics, 1992.

\bibitem{Fleischer:1992}
R.~Fleischer, K.~Mehlhorn, G.~Rote, E.~Welzl, and C.~Yap.
\newblock Simultaneous inner and outer approximation of shapes.
\newblock {\em Algorithmica}, 8(1-6):365--389, Dec. 1992.

\bibitem{gsch}
M.~T. Goodrich and J.~S. Snoeyink.
\newblock Stabbing parallel segments with a convex polygon.
\newblock {\em Computer Vision, Graphics, and Image Processing},
  49(2):152--170, 1990.

\bibitem{gls-si15dt-10}
C.~Gray, M.~L{\"o}ffler, and R.~I. Silveira.
\newblock Smoothing imprecise 1.5d terrains.
\newblock {\em International Journal of Computational Geometry and
  Applications}, 20(4):381--414, 2010.

\bibitem{jin}
K.~{Jin}.
\newblock {Maximal Area Triangles in a Convex Polygon}.
\newblock {\em ArXiv e-prints}, July 2017.

\bibitem{jlp-gcip-11}
A.~J{\o}rgensen, M.~L{\"o}ffler, and J.~Phillips.
\newblock Geometric computations on indecisive points.
\newblock In {\em 12th Algorithms and Data Structures Symposium}, LNCS 6844,
  pages 536--547, 2011.

\bibitem{jup}
W.~Ju, J.~Luo, B.~Zhu, and O.~Daescu.
\newblock Largest area convex hull of imprecise data based on axis-aligned
  squares.
\newblock {\em J. of Combinatorial Optimization}, 26(4):832--859, 2013.

\bibitem{kallus}
Y.~{Kallus}.
\newblock {A linear-time algorithm for the maximum-area inscribed triangle in a
  convex polygon}.
\newblock {\em ArXiv e-prints}, June 2017.

\bibitem{kluv}
V.~Keikha, M.~L{\"{o}}ffler, A.~Mohades, J.~Urhausen, and I.~van~der Hoog.
\newblock Maximum-area triangle in a convex polygon, revisited.
\newblock {\em CoRR}, abs/1705.11035, 2017.

\bibitem{kgonsolved}
V.~Keikha, M.~van~de Kerkhof, M.~van Kreveld, I.~Kostitsyna, M.~L{\"o}ffler,
  F.~Staals, J.~Urhausen, Y.~Vermulen, and L.~Wiratma.
\newblock {Convex partial transversals of planar regions}.
\newblock {\em ArXiv e-prints.}, July 2017.

\bibitem{kl-gmiphd-09}
H.~Kruger and M.~L{\"o}ffler.
\newblock Geometric measures on imprecise points in higher dimensions.
\newblock In {\em 25th European Workshop on Computational Geometry}, pages
  121--124, 2009.

\bibitem{lofflerphdthesis}
M.~L{\"o}ffler.
\newblock {\em Data imprecision in computational geometry}.
\newblock PhD thesis, Utrecht Univesity, 2009.

\bibitem{39}
M.~L{\"o}ffler and M.~van Kreveld.
\newblock Largest and smallest convex hulls for imprecise points.
\newblock {\em Algorithmica}, 56(2):235--269, 2010.

\bibitem{diam}
M.~L{\"o}ffler and M.~van Kreveld.
\newblock Largest bounding box, smallest diameter, and related problems on
  imprecise points.
\newblock {\em Computational Geometry}, 43(4):419--433, 2010.

\bibitem{mukhopadhyay2006all}
A.~Mukhopadhyay, S.~Chatterjee, and B.~Lafreniere.
\newblock On the all-farthest-segments problem for a planar set of points.
\newblock {\em Information processing letters}, 100(3):120--123, 2006.

\bibitem{mukhopadhyay2013all}
A.~Mukhopadhyay and S.~C. Panigrahi.
\newblock All-maximum and all-minimum problems under some measures.
\newblock {\em J. of Discrete Algorithms}, 21:18--31, 2013.

\bibitem{Myers:2010:UGD:1839778.1839801}
Y.~Myers and L.~Joskowicz.
\newblock Uncertain geometry with dependencies.
\newblock In {\em 14th ACM Symposium on Solid and Physical Modeling}, SPM '10,
  pages 159--164, 2010.

\bibitem{nt-teb-00}
T.~Nagai and N.~Tokura.
\newblock Tight error bounds of geometric problems on convex objects with
  imprecise coordinates.
\newblock In {\em Japanese Conference on Discrete and Computational Geometry},
  LNCS 2098, pages 252--263, 2000.

\bibitem{43}
M.~Overmars.
\newblock Finding sets of points without empty convex 6-gons.
\newblock {\em Discrete \& Computational Geometry}, 29(1):153--158, 2003.

\bibitem{50}
M.~H. Overmars.
\newblock {\em Finding minimum area k-gons}, volume~89.
\newblock Unknown Publisher, 1989.

\bibitem{preparata1985computational}
F.~P. Preparata and M.~I. Shamos.
\newblock Computational geometry: An introduction (monographs in computer
  science).
\newblock {\em Springer}, 1:2, 1985.

\bibitem{Salesin}
D.~Salesin, J.~Stolfi, and L.~Guibas.
\newblock Epsilon geometry: Building robust algorithms from imprecise
  computations.
\newblock In {\em 5th Annual Symposium on Computational Geometry}, 1989.

\bibitem{shd}
M.~I. Shamos.
\newblock {\em Problems in computational geometry}.
\newblock unpublished manuscript, 1975.

\bibitem{farnaz}
F.~Sheikhi, A.~Mohades, M.~de~Berg, and A.~D. Mehrabi.
\newblock Separability of imprecise points.
\newblock {\em Computational Geometry}, 61:24--37, 2017.

\bibitem{surimost}
S.~Suri, K.~Verbeek, and H.~Y{\i}ld{\i}z.
\newblock On the most likely convex hull of uncertain points.
\newblock In {\em European Symposium on Algorithms}, pages 791--802. Springer,
  2013.

\bibitem{tasdizen2003feature}
T.~Tasdizen and R.~Whitaker.
\newblock Feature preserving variational smoothing of terrain data.
\newblock In {\em 2th IEEE Workshop on Variational, Geometric and Level Set
  Methods in Computer}, 2003.

\end{thebibliography}








\end{document}